\documentclass[]{article} 

\usepackage{amsmath,amssymb,amsthm,amsfonts}
\usepackage{times}
\usepackage{graphicx}
\usepackage{color}
\usepackage{bbm}
\usepackage{latexsym}
\usepackage{dsfont}
\usepackage{verbatim}
\usepackage{booktabs,array}
\usepackage{authblk}
\usepackage{tikz}

\usepackage{hyperref}
\definecolor{bl}{RGB}{30,30,150}
\hypersetup{colorlinks=true, citecolor=bl, linkcolor=bl, urlcolor=bl}

\usetikzlibrary{tikzmark,calc,,arrows,shapes,decorations.pathreplacing}
\tikzset{every picture/.style={remember picture}}
\usetikzlibrary{fit,shapes.misc}
\usetikzlibrary{positioning,backgrounds}

\newtheorem{theorem}{Theorem}
\newtheorem{corollary}[theorem]{Corollary}
\newtheorem{lemma}[theorem]{Lemma}
\newtheorem{proposition}[theorem]{Proposition}
\theoremstyle{definition}
\newtheorem{definition}[theorem]{Definition}

\newtheorem{examplex}[theorem]{Example}
\newtheorem{remarkx}[theorem]{Remark}

\newenvironment{example}
  {\pushQED{\qed}\examplex}
  {\popQED\endexamplex}
  
\newenvironment{remark}
  {\pushQED{\qed}\remarkx}
  {\popQED\endremarkx}
 
\def\a{\alpha } \def\b{\beta } \def\d{\delta } \def\D{\Delta } \def\e{\epsilon } \def\g{\gamma }   \def\l{\lambda } \def\L{\Lambda }       

  \newcommand{\sC}{{\cal C}}   \newcommand{\sF}{{\cal F}}  \newcommand{\sM}{{\cal M}}     \newcommand{\sL}{{\cal L}}    \newcommand{\sY}{{\cal Y}}  \newcommand{\sS}{{\cal S}} \newcommand{\sX}{{\cal X}}
\newcommand{\bX}{\mathbf{X}}
\newcommand{\bY}{\mathbf{Y}}
\newcommand{\bx}{\mathbf{x}}
\newcommand{\by}{\mathbf{y}}

    \newcommand{\R}{{\mathbb R}}

\newcommand{\supp}{\operatorname{supp}}

\newcommand{\bq}{\begin{equation}}
\newcommand{\eq}{\end{equation}}
\newcommand{\bpm}{\begin{pmatrix}}
\newcommand{\epm}{\end{pmatrix}}  
\begin{document}

\title{\Large\bf Factorized Mutual Information Maximization\bigskip}
\author[1]{{\normalsize\bf Thomas Merkh}\thanks{tmerkh@math.ucla.edu}}
\author[1,2]{{\normalsize\bf Guido Mont\'ufar}\thanks{montufar@math.ucla.edu}}

\affil[1]{\small Department of Mathematics, University of California, Los Angeles, CA 90095, USA}
\affil[2]{\small Department of Statistics, University of California, Los Angeles, CA 90095, USA}
\affil[2]{\small Max Planck Institute for Mathematics in the Sciences, 04103 Leipzig, Germany}
\date{\small\today}
\maketitle
    
\begin{abstract}
We investigate the sets of joint probability distributions that maximize the average multi-information over a collection of margins. These functionals serve as proxies for maximizing the multi-information of a set of variables or the mutual information of two subsets of variables, at a lower computation and estimation complexity. We describe the maximizers and their relations to the maximizers of the multi-information and the mutual information. 

\end{abstract}

\section{Introduction}

Mutual information (MI) is a measure of mutual dependence between two random variables that plays a central role in information theory and machine learning. The multi-information is a generalization to composite systems with an arbitrary number of random variables. 
The problem of maximizing the multi-information was proposed and studied in~\cite{10.2307/2692015}, motivated by infomax principles. 
It was shown, in the setting of finite valued random variables, that the optimizers of the unconstrained maximization problem are attained within low dimensional exponential families of joint probability distributions. 
A characterization of the optimizers was obtained in~\cite{37261}. 
Maximizing the multi-information can be regarded as a special instance of the more general problem of maximizing a divergence from an exponential family. 
Indeed, the multi-information of a joint probability distribution is equal to its Kullback-Leibler divergence from an independence model. The problem of maximizing the divergence from an exponential family has been advanced by 
Mat\'u\v{s}~\cite{matus2004,Matus:2009:DFD:1720558.1720561}, 
Mat\'u\v{s} and Ay~\cite{37278}, 
Mat\'u\v{s} and Rauh~\cite{6034269}, 
Rauh~\cite{46843}. 
The divergence maximization problem can also be asked for more general classes of probability models, such as probabilistic graphical models with hidden variables and stochastic neural networks. Works in this direction include~\cite{montufar2011expressive,montufar2014universal} and the short overview~\cite{62700}. 

Numerous machine learning applications involve optimizing the MI over a restricted set of joint probability distributions. Examples include the 
information bottleneck methods~\cite{witsenhausen1975,InformationBottleneck,45903}, 
the computation of positive information decompositions~\cite{e16042161}, 
the implementation of information theoretic regularizers in robotics and reinforcement learning~\cite{doi:10.1177/1059712310375314,montufar2016information}, 
the analysis of deep networks~\cite{NIPS2018_7453}, 
and unsupervised representation learning in deep learning~\cite{hjelm2018learning}. 
A parametrized set of distributions may or may not contain the unconstrained maximizers of a given functional. In some cases, the distributions are constrained by structures in addition to the chosen parametrization. For instance, in Markov decision processes the joint distributions between time consecutive states are governed not only by a parametrized policy model but also by the state transition mechanism; see, e.g.,~\cite{10.1371/journal.pcbi.1004427}. 
In order to optimize the mutual information in such cases, we usually need to resort to iterative parameter optimization techniques. Computing the parameter updates can quickly become intractable, even for a moderate number of variables, since the number of possible joint states grows exponentially in the number of variables. 
The estimation and computation of the mutual information is often a bottleneck in the implementation of algorithms that are based on it. 
For reference, the estimation of the MI from observations has been studied in~\cite{roulston1999estimating,kraskov2004estimating, slonim2005estimating,gao2015efficient}, 
and recently also using neural networks in~\cite{pmlr-v80-belghazi18a}. 

We are interested in proxies that are easier to estimate from samples and easier to compute than a given functional, and which might provide a signal for optimizing it, possibly having the same optimizers, or a structured subset of the optimizers. 
We consider measures defined as averages of the multi-information over subsets of all random variables in a composite system. 
Based on previous investigations of the maximizers of multi-information~\cite{37261, 10.2307/2692015}, we can expect that the maximizers of these measures will satisfy several properties. 
First, they may be contained within the closure of a low dimensional exponential family. Second, like many other instances of divergence maximizing distributions, they may exhibit reduced support. Last, if the measures are symmetric under the exchange of indices, the maximizing distributions are expected to display the same degree of symmetry. 
We characterize the sets of maximizers of various factorized measures and describe their relation to the maximizers of the multi-information and the maximizers of the mutual information of two subsets of variables. 

As a motivating application, we have in mind the estimation and maximization of the mutual information between time consecutive sensor readings of a reinforcement learning agent acting in an environment. The mutual information is the $1$ time step version of a notion called predictive information, which has been considered in the literature~\cite{bialek2001predictability,schossau2015information,ay2008predictive,zahedi2013linear,crutchfield2001synchronizing}. In~\cite{montufar2016information}, a factorized mutual information (here called SFMI) was used in place of the mutual information as an intrinsic reward signal to encourage more robust walking behaviors for a multi-legged robot. A similar approach was also used previously in~\cite{doi:10.1177/1059712310375314} for a chain of robots. Maximizing the mutual information between time consecutive measurements of a given variable should encourage behaviors that are diverse and, at the same time, predictable. This intuition comes from writing the mutual information between $X,Y$ as $I(X,Y) = H(X) - H(X|Y)$, where $H$ is the entropy. 
In this context, the mutual information serves as a regularizer that is added in order to ease the task-objective optimization problem and also to assign preferences to specific types of solutions. 
Some references to such regularizes, also known as intrinsic motivation, include~\cite{baldassarre2013intrinsically,mohamed2015variational,still2012information,chentanez2005intrinsically,klyubin2005empowerment,bialek2001predictability,pathak18largescale}. 

\medskip

This article is organized as follows. 
In Section~\ref{sec:2} we review definitions and in Section~\ref{sec:2a} relevant existing results around the multi-information. 
In Section~\ref{sec:3} we define notions of factorized multi and mutual information. 
The FMI of a set of random variables is defined as an average multi-information of subsets of variables. 
The SFMI of two random vectors is defined as an average mutual information of non-overlapping pairs with one variable from each vector. 
In Section~\ref{sec:4} we study the maximizers of the FMI and characterize the cases when they coincide with the maximizers of the multi-information. 
In Section~\ref{sec:5} we characterize the maximizers of the SFMI, which build a union of transportation polytopes that contain all maximizers of the multi-information and some maximizers of the mutual information. 
Section~\ref{sec:codes} describes partitions of sets of strings that we use for the characterization of SFMI maximizers. 
In Section~\ref{sec:6} we offer a summary and discussion of our results.

\section{Multi-information and mutual information}
\label{sec:2}

We consider discrete probability distributions supported on a finite set $\sX$. 
The set of all such distributions is a $(|\sX|-1)$-simplex denoted $\Delta_\sX$. 
Let $D(p \| q)$ be the Kullback-Leibler divergence between the probability distributions $p$ and $q$, which is defined as 
    \begin{equation} \label{eq:1}
    D(p\|q) := \sum_{x\in\sX} p(x) \log{\bigg(\frac{p(x)}{q(x)}\bigg)}.
    \eq 
    If ${\cal C}$ is a family of probability distributions, then let $D(p \| \sC)$ denote the infimum divergence between $p$ and any distribution in $\sC$, 
    \bq \label{eq:2}
    D(p \| {\cal C}) := \inf_{q \in {\cal C}} D(p \| q).
    \end{equation}
    We are concerned with joint probability distributions of $n$ variables, so that $x=(x_1,\ldots, x_n)$ and $\sX=\sX_1\times\cdots\times\sX_n$. 
    We denote $\sF$ the set of fully factorizable joint distributions, i.e., those $p\in\Delta_\sX$ which can be written as 
    \bq \label{eq:3}
p(X_1,\dots,X_n) = p_1(X_1)p_2(X_2) \cdots p_n(X_n), \quad p_i\in\Delta_{\sX_i},\; i=1,\ldots, n. % \big\}.
    \eq
    The set $\mathcal{F}$ is also referred to as the independence model of $X_1,\ldots, X_n$. It is a $\sum_{i=1}^n(|\sX_i|-1)$-dimensional manifold in $\Delta_\sX$. 
    The multi-information of $n$ random variables with joint distribution $p$ is then defined as 
    \bq \label{eq:4}
    I_p(X_1,X_2, \dots, X_n) := D(p(X_1,X_2,\dots,X_n) \| \sF). 
    \eq
    It has a natural interpretation as the distance of $p$ from being independent. 
    Following standard notation we drop the subscript $p$, 
    and write simply $I(X_1,\ldots, X_n)$.  
    The minimum Kullback-Leibler divergence from a joint distribution $p$ to the set of factorizable distributions is attained uniquely by the product of its marginals, 
    \begin{equation}
    \operatorname{arginf}_{q\in \mathcal{F}} D(p\|q) = p(X_1)\cdots p(X_n), 
    \end{equation}
    where $p(X_i=x_i) = \sum_{x_j\colon j\neq i} p(x_1,\ldots, x_n)$, $x_i\in\mathcal{X}_i$, for $i=1,\ldots,n$. 
    The product of marginals is the maximum likelihood estimate of a joint distribution as a product distribution. 
    
    The multi-information can be written equivalently in terms of entropies, as %(see Amari, or Ay)
    \bq \label{eq:4b}
    I(X_1,X_2, \dots, X_n) 
    = \sum_{i=1}^n H(X_i) - H(X_1,\ldots, X_n). 
    \eq
    Here $H(X) = -\sum_{x}p(x)\log{p(x)}$ denotes the entropy of $X$. 
    If $|\sX| = N$, then direct computation shows that the entropy is bounded above by $\log{(N)}$. 
    Therefore, the maximum value of multi-information of $n$ $N$-ary variables is bounded by $n \log{(N)}$. 
    However, as noted in~\cite[Lemma 4.1]{37261}, this bound is never attained, with a sharp bound being 
    \bq
     I(X_1,X_2, \dots, X_n) \leq (n-1)\log{(N)}.
    \eq    
    The set of all distributions which maximize multi-information for $n$ variables will be denoted by $\sM_n,$ or $\sM$ when $n$ is understood.  We note the chain rule of entropy,  
    \bq
    H(X_1,\ldots, X_n) = \sum_{i=1}^n H(X_i|X_1,\ldots, X_{i-1}). 
    \eq
    From this, we see that maximizing the multi-information~\eqref{eq:4b} can be interpreted as maximizing the marginal entropies while at the same time minimizing the conditional entropies. 
    In turn, the maximizers of the multi-information are joint distributions where each individual variable is diverse (high entropy), but predictable from the values of the other variables (low conditional entropy). 
    
    \sloppy 
    The multi-information of two variables is called mutual information. We will consider the mutual information of two random vectors $\bX=(X_1,\ldots, X_n)$ and $\bY=(Y_1,\ldots, Y_n)$, which is given by 
    \begin{equation}
    MI(\bX,\bY) := D(p(\bX,\bY)\|p(\bX)p(\bY)), 
    \end{equation}
    where $p(\bX=\bx) = \sum_{\by}p(\bx,\by)$ and $p(\bY=\by) = \sum_{\bx}p(\bx,\by)$ are the marginal distributions. 
    Note that here the divergence is measured to the set of distributions that factorize as a product $p(\bX)p(\bY)$, which does not need to factorize further over the components of each random vector.

\section{Maximizers of the multi-information} 
\label{sec:2a}

In this section we collect relevant existing results on the multi-information that we will use in the following sections of this article. 

\begin{proposition}[Maximizers of the multi-information,~{\cite[Corollary~4.10]{37261}}]
    \label{prop:MI}
        Consider $N$-ary random variables $X_1,\ldots, X_n$. 
        Then $\max_{p\in\Delta_\mathcal{X}} I_p(X_1,\ldots, X_n) = (n - 1)\log(N)$, and the maximizers are the uniform distributions on 
        $N$-ary codes of length $n$, minimum Hamming distance $n$, and cardinality $N$. The number of such codes is $(N!)^{n-1}$. 
\end{proposition}

The work~\cite{37261} also characterizes the maximizers of the multi-information for variables with different state spaces, which is more technical. We will focus on the case where all variables are $N$-ary. 

\begin{example}
In the case of binary variables, $N=2$, Proposition~\ref{prop:MI} states that the set of maximizers of the multi-information consists of all uniform distributions over pairs of binary vectors of Hamming distance $n$. 
Letting $\pi_i : \{0,1\} \to \{0,1\}$ be one-to-one maps for $i = 2,3,\dots, n$, the maximizers of the multi-information take the form 
\begin{equation*}
%\label{eq:9}
        p%(x_1,\dots,x_n) 
        = \frac12 \left( \d_{(0\pi_2(0)\pi_3(0)\cdots \pi_n(0))} + \d_{(1\pi_2(1)\pi_3(1)\cdots \pi_n(1))} \right),  
\end{equation*}
where $\delta_x$ denotes the point measure supported on string $x$. 
\end{example}

The maximizers of the mutual information are obtained from applying Proposition~\ref{prop:MI} to the special case of two random variables. 
\begin{corollary} 
\label{cor:one}
For two $N$-ary variables $X$ and $Y$, the maximum mutual information is $\max_{p\in\Delta_\mathcal{X}} MI_p(X,Y) = \log(N)$, and the maximizers are the uniform distributions on $N$-ary codes of length $2$, minimum Hamming distance $2$, and cardinality~$N$. The number of such codes is $N!$. 
\end{corollary}
We can apply Corollary~\ref{cor:one} not only to pairs of variables, but also to pairs of random vectors, since random vectors can be regarded as random variables with states over a product space. 

\begin{corollary}
\label{cor:MI}
For two $N$-ary random vectors $\bX=(X_1,\ldots, X_n)$ and $\bY=(Y_1,\ldots,Y_n)$, the maximizers of $MI(\bX,\bY)$ are the uniform distributions over $N$-ary codes of length $2n$, minimum distance $2$ when viewed as $N^n$-ary codes of length $2$, and cardinality $N^n$. The number of such codes is $N^n!$. 
\end{corollary}

\begin{example} \label{example:four}
Consider four binary random variables $(X_1,X_2,Y_1,Y_2)$; that is, $n=2$ pairs of $N=2$ valued variables. 
The $N!^{2n-1}=8$ maximizers of the multi-information of $(X_1,X_2,Y_1,Y_2)$ are 
    \begin{align*}
    &\tfrac{1}{2}(\delta_{0000}+\delta_{1111}),\quad
    \tfrac{1}{2}(\delta_{0101}+\delta_{1010}),\\
    &\tfrac{1}{2}(\delta_{0001}+\delta_{1110}),\quad
    \tfrac{1}{2}(\delta_{0100}+\delta_{1011}),\\
    &\tfrac{1}{2}(\delta_{0010}+\delta_{1101}),\quad
    \tfrac{1}{2}(\delta_{0111}+\delta_{1000}),\\
    &\tfrac{1}{2}(\delta_{0011}+\delta_{1100}),\quad
    \tfrac{1}{2}(\delta_{0110}+\delta_{1001}). 
    \end{align*}

The $N^n!=24$ maximizers of the mutual information of $\bX=(X_1,X_2)$ and $\bY=(Y_1,Y_2)$ are
    \begin{align*}
    \tfrac{1}{4}(\delta_{0010}+\delta_{0100}+\delta_{1001}+\delta_{1111}),\quad
    &\tfrac{1}{4}(\delta_{0000}+\delta_{0110}+\delta_{1001}+\delta_{1111}),^{\ast\ast}\\
    \tfrac{1}{4}(\delta_{0001}+\delta_{0100}+\delta_{1010}+\delta_{1111}),\quad
    &\tfrac{1}{4}(\delta_{0010}+\delta_{0100}+\delta_{1011}+\delta_{1101}),^{\ast\ast}\\
    \tfrac{1}{4}(\delta_{0001}+\delta_{0110}+\delta_{1000}+\delta_{1111}),\quad
    &\tfrac{1}{4}(\delta_{0011}+\delta_{0101}+\delta_{1010}+\delta_{1100}),^{\ast\ast}\\
    \tfrac{1}{4}(\delta_{0000}+\delta_{0101}+\delta_{1011}+\delta_{1110}),\quad
    &\tfrac{1}{4}(\delta_{0001}+\delta_{0111}+\delta_{1000}+\delta_{1110}),^{\ast\ast}\\
    \tfrac{1}{4}(\delta_{0000}+\delta_{0111}+\delta_{1001}+\delta_{1110}),\quad
    &\tfrac{1}{4}(\delta_{0000}+\delta_{0101}+\delta_{1010}+\delta_{1111}),^{\ast}\\
    \tfrac{1}{4}(\delta_{0011}+\delta_{0110}+\delta_{1000}+\delta_{1101}),\quad
    &\tfrac{1}{4}(\delta_{0010}+\delta_{0111}+\delta_{1000}+\delta_{1101}),^{\ast}\\
    \tfrac{1}{4}(\delta_{0000}+\delta_{0110}+\delta_{1011}+\delta_{1101}),\quad
    &\tfrac{1}{4}(\delta_{0011}+\delta_{0110}+\delta_{1001}+\delta_{1100}),^{\ast}\\
    \tfrac{1}{4}(\delta_{0000}+\delta_{0111}+\delta_{1010}+\delta_{1101}),\quad
    &\tfrac{1}{4}(\delta_{0001}+\delta_{0100}+\delta_{1011}+\delta_{1110}),^{\ast}\\
    \tfrac{1}{4}(\delta_{0011}+\delta_{0101}+\delta_{1000}+\delta_{1110}),\quad
    &\tfrac{1}{4}(\delta_{0010}+\delta_{0101}+\delta_{1000}+\delta_{1111}),\\
    \tfrac{1}{4}(\delta_{0001}+\delta_{0110}+\delta_{1011}+\delta_{1100}),\quad
    &\tfrac{1}{4}(\delta_{0010}+\delta_{0101}+\delta_{1011}+\delta_{1100}),\\
    \tfrac{1}{4}(\delta_{0001}+\delta_{0111}+\delta_{1010}+\delta_{1100}),\quad
    &\tfrac{1}{4}(\delta_{0011}+\delta_{0100}+\delta_{1010}+\delta_{1101}),\\
    \tfrac{1}{4}(\delta_{0010}+\delta_{0111}+\delta_{1001}+\delta_{1100}),\quad
    &\tfrac{1}{4}(\delta_{0011}+\delta_{0100}+\delta_{1001}+\delta_{1110}). 
    \end{align*} 

Some maximizers of $MI(\bX,\bY)$ can be expressed as convex combinations of maximizers of $I(X_1,\ldots, X_n,Y_1,\ldots, Y_n)$. 
In the current example, those marked by $\ast$ or $\ast\ast$ can. The four distributions marked with ${\ast}$ have marginals $p(X_1,Y_1)$ and $p(X_2,Y_2)$ of maximum mutual information. 
The four distributions marked with ${\ast\ast}$ have margins $p(X_1,Y_2)$ and $p(X_2,Y_1)$ of maximum mutual information. 
We will discuss them further in Example~\ref{ex:1}. 
\end{example}

\section{Factorized Measures of Multi-Information}\label{sec:3}

We propose alternatives to the multi-information based on averages over subsets of random variables. After presenting the definitions in this section, we investigate the maximizers in the next Sections~\ref{sec:4} and~\ref{sec:5}. 

\begin{definition}
For a family $\Lambda$ of subsets of $\{1,\ldots, n\}$, the \emph{$\Lambda$-factorized multi-information} of $(X_1,\ldots, X_n)$ is defined as 
\begin{equation}
I_\Lambda(X_1,\ldots, X_n) := \frac{1}{|\Lambda|} \sum_{\lambda\in\Lambda} I((X_i)_{i\in\lambda}). 
\end{equation}
Here $I((X_i)_{i\in\lambda})$ is the multi-information of the marginal distribution of variables $X_i$, $i\in\lambda$. 
In particular, $I_{\{1,\ldots,n\}}\equiv I$ and $I_{\{i\}}\equiv 0$. 
\end{definition}

We will focus on two special cases. 
The first is the average mutual information between all pairs of variables. 
We call this the average mutual information, or the pairwise factorized multi-information (FMI). 

\begin{definition}[Factorized multi-information, FMI]
For $n$ random variables $X_1,\ldots, X_n$, the FMI is 
    \begin{equation} \label{eq:5}
    FMI(X_1,\dots,X_n) := \frac{1}{{n\choose 2}} \sum_{i\neq j} MI(X_i,X_j).
    \end{equation}
\end{definition}

The second case that we consider is the average mutual information between pairs of variables with one variable from each of two random vectors. 
We call this the separated average mutual information, or the separated factorized mutual information (SFMI).
    
\begin{definition}[Separated factorized mutual information, SFMI]
For two random vectors ${\bf X} = (X_1,\dots,X_n)$ 
    and ${\bf Y} = (Y_1,\dots,Y_n)$, 
    the SFMI is 
    \bq \label{eq:7}
        SFMI({\bf X},{\bf Y}) := \frac{1}{n} \sum_{i=1}^n MI(X_i, Y_i). 
    \eq
\end{definition}

In the reinforcement learning application mentioned in the introduction, $(X_i,Y_i)$ corresponds to time consecutive readings of the $i$th sensor of the agent. The average over all sensors, $i=1,\ldots, n$, is used as a proxy for $MI(\bX,\bY)$. 
We will investigate the properties of this choice. 

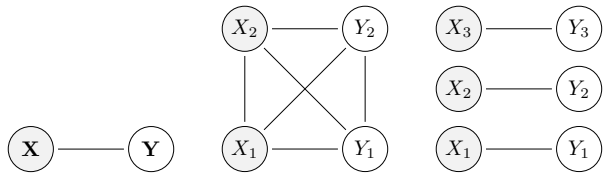
\begin{figure}%[b]
    \centering
    \scalebox{.8}{
    \begin{tikzpicture}
    \node[circle, fill = gray!10, inner sep=0pt, minimum size=.75cm, draw=black, label=center:${\bf X}$] (X) at (0,0) {};
    \node[circle, fill = gray!0, inner sep=0pt, minimum size=.75cm, draw=black, label=center:${\bf Y}$] (Y) at (2,0) {};
    
    \draw[-,shorten >= 2pt, shorten <= 2pt] (X) -- (Y) node[midway,right]{};
    \end{tikzpicture}
\qquad
    \begin{tikzpicture}
    \node[circle, fill = gray!10, inner sep=0pt, minimum size=.75cm, draw=black, label=center:$X_{1}$] (X1) at (0,1) {};
    \node[circle, fill = gray!0, inner sep=0pt, minimum size=.75cm, draw=black, label=center:$Y_{1}$] (Y1) at (2,1) {};
    \node[circle, fill = gray!10, inner sep=0pt, minimum size=.75cm, draw=black, label=center:$X_{2}$] (X2) at (0,3) {};
    \node[circle, fill = gray!0, inner sep=0pt, minimum size=.75cm, draw=black, label=center:$Y_{2}$] (Y2) at (2,3) {};
    \foreach \x in {1,2}{
    \foreach \y in {1,2}{ 
    \draw[-,shorten >= 2pt, shorten <= 2pt] (X\x) -- (Y\y) node[midway,right]{};}}
    \draw[-,shorten >= 2pt, shorten <= 2pt] (X1) -- (X2) node[midway,right]{};
    \draw[-,shorten >= 2pt, shorten <= 2pt] (Y1) -- (Y2) node[midway,right]{};
    \end{tikzpicture}
\qquad
    \begin{tikzpicture}
    \foreach \x in {1,2,3}{ 
    \node[circle, fill = gray!10, inner sep=0pt, minimum size=.75cm, draw=black, label=center:$X_{\x}$] (X\x) at (0,\x) {};
    \node[circle, fill = gray!0, inner sep=0pt, minimum size=.75cm, draw=black, label=center:$Y_{\x}$] (Y\x) at (2,\x) {};}
    
    \foreach \x in {1,2,3}{ 
    \draw[-,shorten >= 2pt, shorten <= 2pt] (X\x) -- (Y\x) node[midway,right]{};}
    \end{tikzpicture}
}
    \caption{Graphical illustration of MI, FMI, SFMI.  Graphically, the $I_\Lambda$ with sets $\l \in \Lambda$ of cardinality $|\l| > 2$ includes hyperedges, but the maximizing set can be described equivalently in terms of the edges contained in elements of $\Lambda$. 
    }
    \label{fig:1}
\end{figure}

An illustration of the proposed factorized measures is shown in Figure~\ref{fig:1}. In this figure, edges correspond to the pair marginals whose mutual-information is being added in the factorized measure. 
One natural question which arises is: how do the distributions with large multi-information compare to those with large factorized mutual information, under the various choices of margins? 
Since the factorized measures average interdependence between subsets of random variables, they can be less descriptive than measuring the joint interdependence of all of the variables. 
In the case where the joint distribution factorizes into the product of pairwise interaction terms, the multi-information naturally decomposes as the sum over mutual information terms~\cite{friedman2001multivariate}.  When higher order interactions exist, the average of pairwise mutual information can serve as an approximation to the multi-information. 
In relation to this, we note that the exponential family with sufficient statistics computing $\Lambda$ margins is, by definition, able to attain all feasible values of these margins, and in particular it is able to express all possible values of $I_\Lambda$. 

The trade-off between computing averages of mutual information between pairs of variables and multi-information or mutual information between groups of variables is that the former can be more easily estimated, since it involves only pair margins and not the full joint distribution. 
The computation of the pairwise factorized multi-information scales with $n^2 \cdot N^2$, where $N=\max_i|\sX_i|$, whereas the computation of multi-information scales with $N^n$. 

\section{Maximizers of the Factorized Multi-information}\label{sec:4}

The next lemma shows that every marginal of a joint distribution with maximum multi-information also has maximum multi-information. A related statement, for pair margins, appeared previously in~\cite[Remark 4.5 (c)]{37261}. 

\begin{lemma}
    \label{lemma:one}
    Consider $N$-ary variables $X_1,\ldots, X_n$. 
    \sloppy If a distribution maximizes the multi-information of $X_1,\ldots, X_n$, then for any subset of indices $\{j_1,\ldots, j_k\}\subseteq \{1,\ldots,n\}$ the corresponding marginal maximizes the multi-information of $X_{j_1},\ldots, X_{j_k}$. 
\end{lemma}
    
\begin{proof}
Consider a joint distribution $p(X_1,\ldots, X_n)$ which attains the maximum of the multi-information. 
By Proposition~\ref{prop:MI}, this is a uniform distribution over an $N$-ary code of length $n$, minimum Hamming distance $n$, and cardinality $N$. 
In this case, the marginal $p(X_n)$ is the uniform distribution over $\mathcal{X}_n$. 
Moreover, the conditional distributions $p(X_1,\ldots, X_{n-1}|X_n=x_n)$, $x_n\in\mathcal{X}_n$, are $N$ point distributions supported on strings of Hamming distance $n-1$. 

The marginal distribution of $(X_1,\ldots, X_{n-1})$ can be written as 
        \begin{align*}
            p(x_1,\dots,x_{n-1}) &= \sum_{x_n} p(x_1,\dots,x_{n-1},x_n)\\
            &=  \sum_{x_n} p(x_1,\dots,x_{n-1}| x_n)p(x_n)\\
            &= \frac{1}{|\sX_n|} \sum_{x_n} p(x_1,\dots,x_{n-1}| x_n), 
        \end{align*}
which is a uniform distribution over $N$ strings of length $n-1$ and minimum Hamming distance $n-1$, and hence has maximum multi-information. 
\end{proof}

Lemma~\ref{lemma:one} can be used to show the following theorem, which characterizes the choices of $\Lambda$ for which the maximizers of $I_\Lambda$ coincide with those of $I$. 

\begin{definition}
\label{definition:connected_covering}
A family $\Lambda$ of subsets of $\{1,\ldots, n\}$ is a \emph{connected covering} if it contains a list of sets $\lambda_t$, $t=1,\ldots, T$, with 
$\lambda_t \cap (\cup_{s=1}^{t-1} \lambda_s) \neq \emptyset$ for $t=2,\ldots, T$, and $\cup_t \lambda_t = \{1,\ldots, n\}$. 
\end{definition}

\begin{theorem}[Maximizers of FMI vs. I]
\label{thm:FMI}
Consider $N$-ary variables $X_1,\ldots, X_n$ with $N>1$, and a family $\Lambda$ of subsets of $\{1,\ldots, n\}$. 
The set of distributions that maximize $I_\Lambda(X_1,\ldots, X_n) = \frac{1}{|\Lambda|}\sum_{\lambda\in \Lambda} I((X_i)_{i\in\lambda})$ coincides with the set of distributions that maximize $I(X_1,\ldots, X_n)$ if $\Lambda$ is a connected covering of $\{1,\ldots, n\}$, and is a strict superset otherwise. 
\end{theorem}

\begin{proof}[Proof of Theorem~\ref{thm:FMI}] 
    Denote $\mathcal{M}$ and $\mathcal{N}$ the sets of joint distributions that maximize the multi-information $I(X_1,\ldots, X_n)$ and the factorized multi-information $I_\Lambda(X_1,\ldots, X_n)$, respectively.  We prove first that $\mathcal{M}\subseteq\mathcal{N}$. 
    Since the $I_\Lambda$ is a sum of bounded non-negative terms, if each term is independently maximized, then $I_\Lambda$ will be maximized. 
    Consider $p(X_1, \dots, X_n) \in \sM$. 
    Lemma~\ref{lemma:one} shows that each marginal $p(X_\l)$, $\l \in \L$, is a multi-information maximizer.  Therefore, $p\in\mathcal{N}$ and $\mathcal{M}\subseteq\mathcal{N}$. 

\smallskip
    
    We proceed with the \emph{if} statement. 
    Assume that $\Lambda$ is a connected covering of $\{1,\ldots, n\}$. 
    We need to show that $\mathcal{N}\subseteq\mathcal{M}$. 
    Suppose that $p\in\mathcal{N}$. Since $\Lambda$ is a connected covering, there is a list of sets $\lambda_t$, $t=1,\ldots, T$, with $\lambda_t\cap(\cup_s^{t-1}\lambda_s)\neq \emptyset$ for $t=2,\ldots,T$, and $\cup_t\lambda_t=\{1,\ldots, n\}$. 
    By the above arguments, each $\l_t$-marginal has maximum multi-information. 
    This means that each marginal $p(X_{\l_t})$ is one of $N!^{|\lambda_t|-1}$ possible distributions (this is the number of $N$-ary codes of cardinality $N$, minimum distance $|\lambda_t|$, and length $|\lambda_t|$). 
    Each marginal fixes the joint states of $|\lambda_t|$ variables to one of $N$ possible choices. 
    Marginals with overlapping variables need to be compatible on the overlap. 
    For a given choice of the $\lambda_s$-marginals, $s=1,\ldots, t-1$, there is exactly one compatible choice of the $\lambda_t$-marginal. 
    Since each variable appears in at least one marginal and there is a sequence of overlapping marginals that cover all variables, the joint states of all variables are fixed to one of $N$ possible strings which have distance $n$ from each other. 
    Therefore, $p\in\mathcal{M}$ and $\mathcal{N}\subseteq\mathcal{M}$. 

\smallskip
    
    It remains to show the \emph{only if} statement. 
    For this, we need to show that $\mathcal{M}\neq \mathcal{N}$ if $\Lambda$ is not a connected covering. If $\Lambda$ is not a connected covering, then one of the following is true: 
    1) there is an $i\in\{1,\ldots, n\}$ which is not contained in any $\lambda\in\Lambda$, or 
    2) $\Lambda$ consists of two (or more) subfamilies of sets that are mutually disjoint. 

    In the first case, $I_\Lambda$ is independent of the marginal distribution of $X_i$. On the other hand, for any maximizer of the multi-information the marginal of $X_i$ is a uniform distribution. Hence, provided that $N>1$ (which we always assume), $\mathcal{M}\neq\mathcal{N}$. 

    In the second case, let $\lambda_1$ and $\lambda_2$ be two disjoint subsets of $\{1,\ldots, n\}$ which contain all sets from the family $\Lambda$. If $\lambda_1$ or $\lambda_2$ only contains one point $i\in\{1,\ldots, n\}$, then $I_\Lambda$ is independent of the marginal distribution of $X_i$ (since $I(X_i)\equiv 0$), and we are in the same situation discussed in 1). 
    Assume therefore that $\lambda_1$ and $\lambda_2$ each contain at least two points. 
    By definition, $I_\Lambda$ only depends on the values of the marginal distributions $p_1$ of $(X_i)_{i\in\lambda_1}$ and $p_2$ of $(X_i)_{i\in\lambda_2}$. 
    Each choice of the two marginals is realized by a set of joint distributions that is known as a $2$-way transportation polytope. 
    If the margins take positive value over $m_1$ and $m_2$ entries, respectively, then the dimension of this polytope is $(m_1-1)(m_2-1)$. The dimension of such $2$-way transportation polytopes is known in the literature; see~\cite{yemelichev1984polytopes}. We will encounter a generalization later in this paper. 
    For any $p\in\mathcal{M}$, Lemma~\ref{lemma:one} shows that the margins $p_1$ and $p_2$ are maximizers of the multi-information, and so, by Proposition~\ref{prop:MI}, both have positive uniform value over $N$ entries. 
    Since $p\in\mathcal{N}$, $\mathcal{N}$ also contains a polytope of dimension $(N-1)(N-1)>0$. 
    On the other hand, $\mathcal{M}$ is a finite discrete set and has dimension zero. 
    Therefore, $\mathcal{M}\neq\mathcal{N}$. 
    This completes the proof. 
    \end{proof}

Theorem~\ref{thm:FMI} states that 
maximizing the multi-information of sufficiently many margins imposes enough restrictions on the joint distributions to ensure that they also have maximum multi-information. 
Moreover, it characterizes the minimal families of marginals that are sufficient. 

In particular, since the set of pairs of $\{1,\ldots, n\}$ defines an overlapping covering, we have the following corollary. 
\begin{corollary}[Maximizers of the FMI]
    \label{cor:FMI}
The maximizers of the FMI \eqref{eq:5} agree with the maximizers of the multi-information \eqref{eq:4}. 
\end{corollary}
 
In fact, Theorem~\ref{thm:FMI} shows that specifying $n-1$ pair margins suffices when the pairs form a connected covering of $\{1,\dots,n\}$. 
In particular, the exponential family with sufficient statistics computing $n-1$ pair margins contains in its closure all maximizers of the multi-information. 
This relates to the result from~\cite[Theorem~5.1]{37261}, which shows that the closure of an exponential family of $n-1$ pure pair interactions contains all maximizers of the multi-information. 

\begin{example}
\label{ex:threebinary}
Consider the case of three binary variables, $\L = \{\{1,2\},\{2,3\}\}$, and the following two mutual information maximizing marginals: 
        \begin{align} 
            p(X_1,X_2) =& \frac12 (\d_{00} + \d_{11}), \label{eq:exmp1}\\
            p(X_2,X_3) =& \frac12 (\d_{00} + \d_{11}). \label{eq:exmp2}
        \end{align}
These marginals fix the joint states of $(X_1,X_2,X_3)$ in the following way. 
Write the joint distribution as a convex combination of point distributions,  
        \bq
            p(X_1,X_2,X_3) = p_{000}\d_{000} + p_{001}\d_{001} + \cdots + p_{111}\d_{111}.
        \eq
        Eqn.~\eqref{eq:exmp1} imposes following linear constraints on the joint distribution: 
        \begin{align}
        \begin{split}
            &p_{000} + p_{001} = \tfrac12,\quad
            p_{110} + p_{111} = \tfrac12,\\
            &p_{010} + p_{011} = 0,\quad 
            p_{100} + p_{101} = 0.
            \end{split}
        \end{align}
        Eqn.~\eqref{eq:exmp2} imposes following linear constraints: 
        \begin{align}
        \begin{split}
            &p_{000} + p_{010} = \tfrac12,\quad
            p_{101} + p_{111} = \tfrac12,\\
            &p_{001} + p_{011} = 0,\quad
            p_{100} + p_{110} = 0. 
            \end{split}
        \end{align}
Since probabilities are non-negative, the last two constraints in each case imply $p_{010} = p_{011} = p_{100} = p_{101} = 0$ and  $p_{001} = p_{011} = p_{100} = p_{110} = 0$. 
Therefore, it must be that $p_{000} = p_{111} = \frac12$, showing that specifying two overlapping marginal distributions to be maximizers is sufficient to conclude that the joint distribution is also a multi-information maximizer.
\end{example}

When $\Lambda$ is not a connected covering, the set of maximizers of $I_\Lambda$ has dimension larger than zero. In this case, the exponential family with sufficient statistics computing $\Lambda$ margins still contains in its closure one joint distribution that expresses each choice of the margins, and in particular it contains maximizers of $I_\Lambda$, but in general it does not contain all of them. 
We will discuss the setting of not connected coverings in more detail in the next Section~\ref{sec:5}. 
We conclude this section with the following remark. 

\begin{remark}
We observe that, while there always exists a joint distribution each of whose marginals maximizes the multi-information, in general there are compatibility requirements among the marginals. 
For example, choosing $p(X_1,X_2) = \frac12(\d_{00}+\d_{11})$ and $p(X_2,X_3) = \frac12 (\d_{01} + \d_{10})$ means that $p(X_1,X_3)$ must be exactly $\frac12 (\d_{01} + \d_{10})$. 

Verifying that a certain specification of margins is indeed feasible, meaning that there exists a joint distribution with those margins, is known as the margin specification problem. 
Given a family $\Lambda$ of subsets of $\{1,\ldots, n\}$, the set of margins $(p_\lambda)_{\lambda\in\Lambda}$ that are feasible, is known in the literature as the $\Lambda$-marginal polytope, which can be obtained as the convex hull of the vectors 
\begin{equation}
F(x) = (F_{\lambda,\tilde x_\lambda}(x))_{\lambda\in\Lambda, \tilde x_\lambda\in\mathcal{X}_\lambda}, \quad x\in\mathcal{X}, 
\end{equation}
where for each $\lambda$ and $\tilde x_\lambda$ there is a coordinate taking value one if $x_\lambda = \tilde x_\lambda$ and value zero otherwise. 
These vectors form the columns of a matrix of sufficient statistics, computing the $\Lambda$-margins, for the exponential family known as the $\Lambda$-interaction model. 
%The matrix defined by these columns is a sufficient statistics matrix for the hierarchical model with interactions $\Lambda$. 

For a given vector $(p_\lambda)_\lambda$ of $\lambda$-margins, $\lambda\in\Lambda$, the set of joint distributions that express these margins is known as a linear $\Lambda$-family at $p$ (where $p$ is a compatible joint distribution)~\cite{Rauh11:Thesis}, 
or also as the $n$-way transportation polytope $T(p_\lambda,\lambda\in\Lambda)$~\cite{yemelichev1984polytopes}. 
This is obtained by intersecting the simplex of joint probability distributions with the orthogonal complement of the row span of $F$ at $p$. 
We will encounter this type of structures in the next section. More details are provided in Appendix~\ref{appendix:transportation_polytopes}. 
\end{remark}

\section{Maximizers of the Factorized Mutual Information}\label{sec:5}

We have seen that the maximizers of any factorized multi-information measure with connected covering, like the ones shown in Figure~\ref{fig:my_label}, coincide with the  maximizers of the multi-information. 
These distributions are very different from the maximizers of the mutual information of two random vectors of length $n\geq2$. 
If we want to maximize $MI(\bX,\bY)$ by maximizing a factorized measure $I_\Lambda$, then we should choose a $\Lambda$ that is not a connected covering. 

\begin{figure}
    \centering
    \scalebox{.8}{    
    \begin{tikzpicture}
    \foreach \x in {1,2,3}{ 
    \node[circle, fill = gray!10, inner sep=0pt, minimum size=.75cm, draw=black, label=center:$X_{\x}$] (X\x) at (0,\x) {};
    \node[circle, fill = gray!0, inner sep=0pt, minimum size=.75cm, draw=black, label=center:$Y_{\x}$] (Y\x) at (2,\x) {};}
    
    \foreach \x in {1,2,3}{ 
    \foreach \y in {1,2,3}{ 
    \draw[-,shorten >= 2pt, shorten <= 2pt] (X\x) -- (Y\y) node[midway,right]{};}}
    \end{tikzpicture}
\qquad
    \begin{tikzpicture}
    \foreach \x in {1,2,3}{ 
    \node[circle, fill = gray!10, inner sep=0pt, minimum size=.75cm, draw=black, label=center:$X_{\x}$] (X\x) at (0,\x) {};
    \node[circle, fill = gray!0, inner sep=0pt, minimum size=.75cm, draw=black, label=center:$Y_{\x}$] (Y\x) at (2,\x) {};}
    
    \foreach \x in {1,2,3}{ 
    \draw[-,shorten >= 2pt, shorten <= 2pt] (X\x) -- (Y\x) node[midway,right]{};}
    \draw[-,shorten >= 2pt, shorten <= 2pt] (X1) -- (Y2) node[midway,right]{};
    \draw[-,shorten >= 2pt, shorten <= 2pt] (X2) -- (Y3) node[midway,right]{};
    
    \end{tikzpicture}
}
    \caption{At first sight, the factorized measures with these graphs might appear to be closely related to the mutual information $MI({\bf X},{\bf Y})$. However, by Theorem~\ref{thm:FMI} the maximizers coincide with the maximizers of the multi-information. As illustrated by Example~\ref{example:four}, they differ significantly from the maximizers of $MI({\bf X},{\bf Y})$.}
    \label{fig:my_label}
\end{figure}
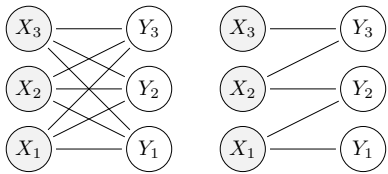

The SFMI \eqref{eq:7} only considers the mutual information of pairs $(X_i,Y_i)$, $i=1,\ldots, n$, which is not a connected covering. 
Theorem~\ref{thm:FMI} shows that the SFMI maximizers are a strict superset of the multi-information maximizers: 
    \begin{corollary} \label{cor:13}
Let $n\geq 2$. The maximizers of $I(X_1,\ldots, X_n,Y_1,\ldots, Y_n)$ are a strict subset of the maximizers of $SFMI(\bX,\bY)$.
    \end{corollary}

In the following we present a characterization of the set of maximizers of SFMI. 
We start with the next theorem, which describes the number of connected components, dimension, and support sets. 

\begin{theorem}[Maximizers of SFMI vs. I] 
\label{theorem:12}
The set of maximizers of the SFMI is a superset of the maximizers of the multi-information. 
It consists of the joint distributions with marginals $p(X_i,Y_i)$ maximizing $MI(X_i,Y_i)$, for $i=1,\ldots, n$. 
This is the disjoint union of $N!^n$ transportation polytopes of dimension $N^n-1 - n(N-1)$ and whose interiors consist of distributions with support of cardinality $N^n$. 
\end{theorem}

The SFMI maximizing polytopes are the sets of joint probability distributions that are compatible with the specification of $n$ pair margins as distributions with maximum mutual information. 
Since the pair margins have disjoint variables, the polytopes can be regarded as orthogonal linear families of an independence model of $n$ $N^2$-ary variables. 

But since the specified margins are uniform distributions with support of cardinality $N$, 
the polytopes have the structure of orthogonal linear families of an independence model of $n$ $N$-ary variables at the uniform distribution. In this perspective, the elementary events of the $i$th variable are pairs $(x_i,y_i)$ forming an $N$-ary code of length $2$, minimum Hamming distance $2$, and cardinality $N$. 

\begin{proof}[Proof of Theorem \ref{theorem:12}]
If we consider a subset $\lambda$ of variables, fixing a marginal distribution $p(X_\lambda)$ with support $\tilde \sX_\lambda\subseteq\sX_\lambda$ defines a linear set in the joint probability simplex of dimension $|\tilde \sX_\lambda|(|\sX_{\lambda^c}|-1)$. This is namely the set of distributions of the form $p(X_\lambda)q(X_{\lambda^c}|X_{\lambda})$. 
The dimension count comes from the fact that for each $x_{\lambda}\in \tilde \sX_\lambda$, the 
conditional distribution $q(X_{\lambda^c}|X_{\lambda}=x_\lambda)$ can be chosen independently as an arbitrary point in a $(|\sX_{\lambda^c}|-1)$ dimensional simplex. 

For variables $(X_1,\ldots, X_n,Y_1,\ldots, Y_n)$, the set of maximizers of the MI of a given pair $(X_i,Y_i)$ corresponds to a set $\sS_i = \cup_{j=1}^{N!} \sS_{i,j}$ in the joint probability simplex which is the disjoint union of $N!$ linear spaces of dimension $N(N^{2n-2}-1)$. 
Indeed, there are $N!$ maximizers of the marginal mutual information. Each of them  corresponds to fixing a marginal distribution $p(X_i, Y_i)$ with support of cardinality $N$, which defines a set $\sS_{i,j}$ in the joint probability simplex. 
In terms of the above discussion, $\sX_\lambda = \sX_i\times \sY_i$ and $|\tilde \sX_\lambda| = N$. 
The set $\mathcal{S}$ of maximizers of the SFMI is the intersection of $n$ such sets, $\sS=\cap_{i=1}^n \,\sS_i$, one corresponding to each pair marginal. 
Each marginal implies that $(N^2-N)N^{n-2}$ entries of the joint distribution vanish. 
Overall, only $N^n$ entries of the joint distribution can have nonzero value. 
Each marginal implies $N$ linear constraints of the form $\frac{1}{N}=p(X_i=x_i,Y_i=y_i)=\sum_{x_{\hat i},y_{\hat i}}p(x_1,\ldots, x_n,y_1,\ldots,y_n)$. 
Each marginal contributes $(N-1)$ independent constraints, and have one shared constraint, namely that the sum of entries is $1$. The number of independent constraints can be derived from the fact that $\mathbf{1}_{\lambda, x_\lambda}$, $\lambda\subseteq\Lambda$, $x_\lambda \in\prod_{i\in\lambda}(\mathcal{X}_i\setminus\{0\})$,  including $\mathbf{1}_\emptyset\equiv1$, is a basis of the set of functions $\sum_{\lambda\in\Lambda}f_\lambda(x)$; see~\cite{Hosten:2002:GBP:636628.636633}. 

We have $N^n$ positive entries satisfying $n(N-1)+1$ independent linear equality constraints, which gives the desired dimension $N^n - n(N-1)-1$. 

In terms of the number of pieces, all pieces defined by one marginal intersect with all pieces defined by any other marginal. Since each pair marginal defines $N!$ pieces, the total number of pieces is $N!^n$. 

The fact that the set of SFMI maximizers is a superset of the set of FMI maximizers follows from the fact that it is defined by a subset of the affine constraints that define the FMI maximizers. 
\end{proof}

The next question is whether the transportation polytopes maximizing the SFMI contain some of the maximizers of the MI.  Transportation polytopes are not very well understood in general.  We obtain the following characterization. 

\begin{theorem}[Maximizers of SFMI vs. MI] \label{theorem:13}
Each of the $N!^n$ transportation polytopes maximizing the SFMI contains $N!^{n-1}$ vertices which are maximizers of the multi-information. 
These vertices come in $(N-1)!^{n-1}$ disjoint sets,
each set consisting of $N^{n-1}$ affinely independent vertices and having the same centroid, which is a maximizer of $MI({\bf X},{\bf Y})$. 
\end{theorem}

To put this in perspective, recall that 
$I(X_1,Y_1,\ldots, X_n,Y_n)$ has $N!^{2n-1}$ maximizers, and $MI(\bX,\bY)$ has $N^n!$ maximizers. 
Moreover, the dimension of each of the transportation polytopes in question is $N^n-1-n(N-1)$. 
The total number of vertices of a transportation polytope is not known in general. 
For the special type of transportation polytopes that we consider here, the theorem gives a lower bound $N!^{n-1}$. 

\begin{proof}[Proof of Theorem~\ref{theorem:13}]
The proof is by construction. 
Consider one of the $N!^n$ polytopes, $\sS = \cap_{i=1}^n \sS_{i,j_i}$ for a fixed choice of $j_i \in\{1,\ldots, N!\}$, for $i=1,\ldots,n$. 
Let $p_{i,j_i}$ be the corresponding $(X_i,Y_i)$ margin, which is a uniform distribution on $N$ strings of length two and minimum distance two, out of the $N^2$ possible strings of length two. 
Denote these $N$ strings by $z^{k}_i = (x^{k}_i,y^{k}_i)$, $k = 1,\ldots, N$. 
In the following we will write the joint states of all variables as $((x_1,y_1),\ldots, (x_n,y_n))$. 

We claim that $\sS$ contains 
$\frac{1}{N}\sum_{k=1}^N \delta_{z_1^{\pi_1(k)}\cdots z_n^{\pi_n(k)}}$, for each choice of an ordering $\pi_i$ of $\{1,\ldots,N\}$, for each $i=1,\ldots, n$. 
Indeed, for each $i=1,\ldots, n$ the $i$th margin is 
$\frac{1}{N}\sum_k \delta_{z_i^k}=p_{i,j_i}$. 

Each of these points is a distinct maximizer of the multi-information of $(Z_1,\ldots, Z_n)$. Indeed, each of them is a uniform distribution on $N$ strings of length $n$ and minimum distance $n$ (with each $Z_i$ regarded as a single $N$ valued variable). 
They are also maximizers of the multi-information of $(X_1,Y_1,\ldots, X_n,Y_n)$, since any distinct $z_i^k$ and $z_i^{k'}$ from our list have, by definition, distance two when regarded as strings $(x^{k}_i,y^{k}_i)$ and $(x^{k'}_i,y^{k'}_i)$. 

Summarizing, $\sS$ contains $N!^{n-1}$ distinct maximizers of the multi-information. These points are also vertices of the polytope. 
Since they have support of cardinality $N$, they cannot be expressed as a non-trivial combination of maximizers of SFMI. Indeed, $N$ is the minimum cardinality of the support of any maximizer of SFMI (because this is the minimum cardinality of any maximizer of the mutual information of any pair margin). 
Moreover, any maximizer of SFMI with support of cardinality $N$ must be the uniform distribution over that support (because this is true for any maximizer of the MI of any pair margin).  

\smallskip

Next we show that $\sS$ contains $(N-1)!^{n-1}$ simplices which are each the convex hull of $N^{n-1}$ distinct vertices of $\sS$. 
Recall that $|\supp(\sS)|=\max\{|\supp(p)|\colon p\in\sS\} = N^n$. 
Because $\sS$ has $N!^{n-1}$ vertices which are multi-information maximizers, it follows that there are $N!^{n-1}$ codes of length $2n$, minimum distance $2n$, and cardinality $N$. These are the support sets of the vertices. 
One can find $N^{n-1}$ codes of this form which are disjoint and cover all strings (see Proposition~\ref{proposition:partition}). 
For $N^{n-1}$ disjoint codes, each of cardinality $N$, the union contains $N^n$ distinct strings and covers $\supp(\sS)$. 
One can find $\frac{N!^{n-1}}{N^{n-1}} = (N-1)!^{n-1}$ such collections of codes 
(see Proposition~\ref{proposition:numberofpartitions}), which correspond to the different simplices. 

Consider the vertices corresponding to a single collection of $N^{n-1}$ vertex support sets.  Any distribution written as a convex combination of these vertices is in $\sS$.  Since the support sets of these vertices are disjoint, it follows that a strictly convex combination will have support on $N^n$ vectors. 
Furthermore, the convex hull of these vertices can be seen as a simplex of dimension $N^{n-1}-1$.
At the center of the simplex where each vertex is given equal weight, one obtains the uniform distribution over the $N^n$ vectors of $\supp(\sS)$. 
Recall from the previous that the vectors in $\supp(\sS)$ are $N$-ary vectors of length $2n$ with minimum distance $2$, and therefore the uniform distribution over $N^n$ such vectors is a maximizer of $MI({\bf X},{\bf Y})$. 
Since each of the $(N-1)!^{n-1}$ simplices have the same support of $\supp(\sS)$, their centroids all coincide at a maximizer of $MI({\bf X},{\bf Y})$. % ok
\end{proof}

In particular, we note the following. 

\begin{corollary}
The maximizers of the multi-information are vertices of the polytopes of maximizers of the SFMI. 
\end{corollary}

We note that the MI is independent of the ordering of the variables within the vectors $\bX$ and $\bY$. 
The SFMI measure that we have introduced corresponds to $I_\Lambda$ where $\Lambda$ is a specific pairing of variables in the two vectors. 
Each $\Lambda$ consisting of $n$ pairs that are a perfect matching between $X_1,\ldots, X_n$ and $Y_1,\ldots, Y_n$ defines a different SFMI measure. 
The situation is illustrated in Figure~\ref{fig:SFMIs}. 
We expect that for each of the $n!$ possible pairings, the corresponding SFMI measure will have among its maximizers, different subsets of maximizers of the MI that can be written as convex combinations of multi-information maximizers. 
At this point it remains an open question how the sets of MI maximizers that maximize different SFMI measures compare to each other, and how many MI maximizes total are maximizers of some SFMI measure. 

We illustrate the discussion of this section in the next example.
    \begin{example}
        \label{ex:1}
        Consider $n=2$ pairs of $N=2$ valued variables $(X_1,X_2,Y_1,Y_2)$. 
        The maximizers of the SFMI~\eqref{eq:7} are $N!^n=4$ polytopes of dimension $N^n-1-n(N-1)=1$. 
        With non-negative $\a,\b$ such that $\a + \b = 1$, the polytopes are 
        \begin{align}
        \begin{split}
            &\tfrac\a2(\d_{0000} + \d_{1111}) + \tfrac\b2(\d_{0101} + \d_{1010}), \label{eq:10}\\
            &\tfrac\a2(\d_{0001} + \d_{1110}) + \tfrac\b2(\d_{0100} + \d_{1011}), \\%\label{eq:11}\\
            &\tfrac\a2(\d_{0010} + \d_{1101}) + \tfrac\b2(\d_{0111} + \d_{1000}), \\%\label{eq:12}\\
            &\tfrac\a2(\d_{0011} + \d_{1100}) + \tfrac\b2(\d_{0110} + \d_{1001}) .%\label{eq:13}.
        \end{split}
        \end{align}
        A detailed derivation is provided in Appendix~\ref{sec:sfmiex}. 
        When $\a,\b$ are strictly positive, these distributions have support of cardinality $4$ and do not maximize multi-information.
        However we see that the vertices, with $\beta=0$ or $\alpha=0$, recover all $8$ maximizers of multi-information. 
        Additionally, setting $\a = \b = \frac12$ recovers $4$ out of the $24$ maximizers of $MI({\bf X},{\bf Y})$. 
        These are the four points marked with ${\ast}$ in Example~\ref{example:four}. 
        Notice that the SFMI maximizers with highest entropy are maximizers of $MI({\bf X},{\bf Y})$. 

\smallskip
 
Consider now an SFMI with a different assignment of $X$s and $Y$s, $SFMI' = \frac12\big(MI(X_1,Y_2) + MI(X_2,Y_1)\big)$.
With non-negative $\a,\b$ such that $\a + \b = 1,$ the polytopes of $SFMI'$ maximizers are
\begin{align}
        \begin{split}
            &\tfrac\a2(\d_{0000} + \d_{1111}) + \tfrac\b2(\d_{1001} + \d_{0110}),\label{eq:21}\\
            &\tfrac\a2(\d_{0010} + \d_{1101}) + \tfrac\b2(\d_{0100} + \d_{1011}),\\
            &\tfrac\a2(\d_{0001} + \d_{1110}) + \tfrac\b2(\d_{1000} + \d_{0111}),\\
            &\tfrac\a2(\d_{0011} + \d_{1100}) + \tfrac\b2(\d_{0101} + \d_{1010}).
        \end{split}
\end{align}
At the center of these polytopes where $\a = \b$, the maximizers marked ${\ast\ast}$ in Example~\ref{example:four} are found, demonstrating that distinct subsets of MI maximizers are contained in each of the $n!$ possible SFMI measures. 
\end{example}

\begin{figure}%[b]
    \centering
    \scalebox{.8}{
    \begin{tikzpicture}
    \foreach \x in {1,2,3}{ 
    \node[circle, fill = gray!10, inner sep=0pt, minimum size=.75cm, draw=black, label=center:$X_{\x}$] (X\x) at (0,\x) {};
    \node[circle, fill = gray!0, inner sep=0pt, minimum size=.75cm, draw=black, label=center:$Y_{\x}$] (Y\x) at (2,\x) {};}
    \foreach \x in {1,2,3}{ 
    \draw[-,shorten >= 2pt, shorten <= 2pt] (X\x) -- (Y\x) node[midway,right]{};}
    \end{tikzpicture}
\qquad
    \begin{tikzpicture}
    \foreach \x in {1,2,3}{ 
    \node[circle, fill = gray!10, inner sep=0pt, minimum size=.75cm, draw=black, label=center:$X_{\x}$] (X\x) at (0,\x) {};
    \node[circle, fill = gray!0, inner sep=0pt, minimum size=.75cm, draw=black, label=center:$Y_{\x}$] (Y\x) at (2,\x) {};}
    
    \draw[-,shorten >= 2pt, shorten <= 2pt] (X1) -- (Y2) node[midway,right]{};
    \draw[-,shorten >= 2pt, shorten <= 2pt] (X2) -- (Y3) node[midway,right]{};
    \draw[-,shorten >= 2pt, shorten <= 2pt] (X3) -- (Y1) node[midway,right]{};
    \end{tikzpicture}
\qquad
   \begin{tikzpicture}
    \foreach \x in {1,2,3}{ 
    \node[circle, fill = gray!10, inner sep=0pt, minimum size=.75cm, draw=black, label=center:$X_{\x}$] (X\x) at (0,\x) {};
    \node[circle, fill = gray!0, inner sep=0pt, minimum size=.75cm, draw=black, label=center:$Y_{\x}$] (Y\x) at (2,\x) {};}
    
    \draw[-,shorten >= 2pt, shorten <= 2pt] (X1) -- (Y3) node[midway,right]{};
    \draw[-,shorten >= 2pt, shorten <= 2pt] (X2) -- (Y1) node[midway,right]{};
    \draw[-,shorten >= 2pt, shorten <= 2pt] (X3) -- (Y2) node[midway,right]{};
    \end{tikzpicture}
}
    \caption{Illustration of various types of SFMI measures. }
    \label{fig:SFMIs}
\end{figure}
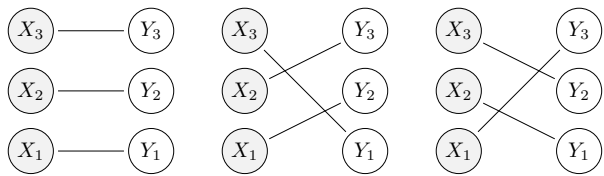    

We provide more examples in Appendix~\ref{sec:sfmiex}. 

\section{Codes and partitions}
\label{sec:codes}

In the proof of Theorem~\ref{theorem:13} we used certain properties of the set of $N$-ary strings of length $n$; specifically, we used that it can be partitioned into codes of minimum distance $n$ and cardinality $N$, in a number of different ways. 
In this section we state and prove these facts. 

\begin{proposition}
There are $N!^{n-1}$ $N$-ary codes of length $n$, minimum distance $n$, and cardinality $N$. 
\end{proposition}
\begin{proof}
An $N$-ary code of length $n$, minimum distance $n$, and cardinality $N$ corresponds to a perfect covering of a rectangular grid of size $N\times n$ by $N$ strings of length $n$. 
In order to count the total number, consider the grid, and the possible choices for how to place the first string. 
The first coordinate is $1$, for the second coordinate there are $N$ choices, for the third coordinate there are $N$ choices, etc. For the second string the first coordinate is $2$, and there are $N-1$ possible choices for each of the remaining coordinates. 
For the $N$th string there is only one choice left. 
So in total we have $N^{n-1} \cdot (N-1)^{n-1}\cdots 1^{n-1} = N!^{n-1}$. 
\end{proof}

\begin{proposition}
\label{proposition:bipartite}
The set of all $N^2$ edges of the full bipartite graph on $N\times N$ can be partitioned into $N$ sets of $N$ edges, each of which forms a perfect matching. 
\end{proposition}

\begin{proof}
%The proof is by construction. 
Denote the edges by $(u,v)\in \{0,\ldots, N-1\}\times\{0,\ldots, N-1\}$. 
Each $t=1,\ldots, N$ defines a perfect matching 
$\{(u,u+t \mod (N))\colon u\in\{0,\ldots, N-1\}\}$. 
Each of these matchings uses one distinct edge, out of the $N$ edges that connect to the $u$-th left node, for each $u$. 
Hence all of these matchings have disjoint edge sets, and they, together, contain all edges from the full bipartite graph. 
\end{proof}

\begin{proposition}
\label{proposition:partition}
The set of all $N^n$ $N$-ary strings of length $n$ 
can be partitioned into $N^{n-1}$ codes of minimum distance $n$ and cardinality $N$. 
\end{proposition}
\begin{proof} 
We want to partition the set of all $N^n$ strings into $N^{n-1}$ sets, each of which gives a perfect covering of a grid of size $N\times n$. 
For any $(t_2,\ldots, t_n)\in \{0,\ldots,N-1\}^{n-1}$, define a code as 
$\{(u,u+t_2 \mod(N), \ldots, u+t_n\mod(N))\colon u=0,\ldots, N-1 \}$. 
The proof of Proposition~\ref{proposition:bipartite} discusses how this construction gives a perfect matching between all nodes in column $i$ and column $i+1$, for all $i=1,\ldots, n-1$. 
There are $N^{n-1}$ of these assignments and they all build disjoint sets of strings. 
They, together, exhaust all $N^n$ strings, by the same reasoning of Proposition~\ref{proposition:bipartite}. 
\end{proof}

\begin{proposition}
\label{proposition:numberofpartitions}
There are $N!^{n-1}/N^{n-1}$ ways to partition the set of all $N^n$ $N$-ary strings of length $n$ in the form described in Proposition~\ref{proposition:partition}. 
\end{proposition}
\begin{proof}
Consider again each $N$-ary code of length $n$, minimum distance $n$, and cardinality $N$ as a perfect covering of an $N\times n$ grid by strings. 
Consider the partition described in Proposition~\ref{proposition:partition} of the set of $N$-ary strings into $N^{n-1}$ such codes. 
Note that this partition is constructed by applying circular shifts, which define a subgroup of the symmetric group. 
Now we simply consider the cosets defined by the circular shifts. 
The cardinality of all cosets is equal. 
Two cosets are either equal or are disjoint. 
The number of distinct cosets is $N!/N = (N-1)!$. 
This gives a total of $(N-1)!^{n-1}$ partitions. 
\end{proof}

\section{Discussion}
\label{sec:6}

We formulated factorized versions of the mutual information and studied the sets of joint probability distributions that maximize them. As we observe, characterizing the maximizers of these measures translates to characterizing the sets of joint distributions that are compatible with a given specification of margins. 

When a sufficient number of margins is included in the factorized measure, the maximizers correspond precisely to the maximizers of the multi-information. For this to be the case, it is necessary and sufficient that the set of margins builds a connected covering of all variables. In particular, there are families of pair margins that are sufficient, which reflects previous results showing that the multi-information can be maximized over families of joint distributions that only include pair interactions. 

When the considered margins have disjoint sets of variables, the constraints that they impose on the joint distributions no longer specify the joint distributions down to the maximizers of the multi-information. Instead, they specify families of transportation polytopes some of whose vertices are multi-information maximizers. As we showed, certain choices of margins (as in the SFMI) lead to polytopes that contain some of the maximizers of the mutual information between two subsets of variables. 

\paragraph{MI vs. I}
The set of maximizers of the multi-information $I(X_1,\ldots, X_n,Y_1,\ldots, Y_n)$ is a discrete set of cardinality $(N!)^{2n-1}$. 
The set of maximizers of the mutual information $MI({\bf X},{\bf Y})$ is a discrete set of cardinality $(N^n)!$. The two sets are equal if $n=1$, and are disjoint otherwise. 
These are well known results, collected in Proposition~\ref{prop:MI} and Corollary~\ref{cor:MI}. 

\paragraph{FMI vs. I and MI}
The set of maximizers of $I_\Lambda(X_1,\ldots, X_n)$ is equal to the set of maximizers of $I(X_1,\ldots, X_n)$ if the considered set $\Lambda$ of margins builds a connected covering of all variables, and it is a strict superset otherwise. See Theorem~\ref{thm:FMI}. 
If $\Lambda$ is a connected covering of all variables, then the maximizers of $I_\Lambda(X_1,\ldots, X_n,Y_1,\ldots, Y_n)$ are disjoint from maximizers of $MI(\bX,\bY)$, unless $n=1$ in which case the two sets are equal. 
If $\Lambda$ is not a connected covering, the relation is open in general. 
We studied in detail the special case of SFMI, summarized below. 

\paragraph{SFMI vs. I and MI}
SFMI is the special case of $I_\Lambda$ with $\Lambda$ being a perfect matching of $X$ and $Y$ variables. This is not a connected covering, except when $n=1$. 
The set of SFMI maximizers consists of $N!^n$ disjoint polytopes of dimension $N^n-1-n(N-1)$, each of which has the structure of a central $n$-way transportation polytope defined by $1$-margins of size $N$. See Theorem~\ref{theorem:12}. 
Each of the $N!^n$ polytopes maximizing SFMI contains $N!^{n-1}$ maximizers of $I(X_1,\ldots, X_n,Y_1,\ldots, Y_n)$ and one maximizer of $MI({\bf X},{\bf Y})$. 
See Theorem~\ref{theorem:13}.

\subsection*{Factorized multi-information as an intrinsic reward}
When the MI is superimposed with another objective function, we may obtain specific types of solutions  which maximize both, or simply which tend to have a larger MI. 

The SFMI maximizers do not include all MI maximizers. 
An interpretation is that the SFMI maximizers also maximize the MI between specific pairs of variables. 
So for instance, the MI can be maximized if $H(Y_1|X_1)\neq 0$, if in exchange one has that $H(Y_2|X_1)=0$, meaning that sensor $2$ can be predicted from sensor $1$. However, the maximizers of SFMI insist that $H(Y_1|X_1)=0$. 
In this sense, the SFMI gives a more structured objective function than the MI. 

\subsection*{Possible generalizations}
We focused on two general types of margins. 
In the case of families that are not connected coverings, we considered the SFMI, which consists of a non overlapping covering by edges. 
Extensions of our analysis to other families of margins are possible. 
Any connected subfamily will then specify a maximizer of the multi-information over the covered variables. 
This will give rise to maximizing sets defined in terms of transportation polytopes with margins given by the maximal connected subfamilies. 
A natural question that arises is how to choose $\Lambda$ in order to capture, as tightly as possible, a given subset of maximizers of $MI(\bX,\bY)$. 

We have focused on systems where all variables have the same number of states. The same problem can be studied in the case of inhomogeneous variables as well. 
One should be able to obtain a characterization as was done in~\cite{37261} for the multi-information.

The case of continuous variables is another geneneralization that would be interesting to consider in the future. This has not been discussed in as much detail in the literature as the discrete case. 
    
\bigskip 
\noindent\textbf{Acknowledgement}
\small 
This project has received funding from the European Research Council (ERC) under the European Union's Horizon 2020 research and innovation programme (grant agreement n\textsuperscript{o} 757983).
    
\bibliographystyle{amsplain}
\bibliography{references}

\appendix 

\section{Related measures of multivariate correlation}
    Multi-information as defined above was first introduced by McGill (1954) \cite{mcgill1954multivariate} under the name of \textit{total correlation}, and later discussed in \cite{watanabe1960information} under the same name.  Over time as the relationship between total correlation and mutual information became clear, several authors have begun to call this quantity multi-information \cite{37261, friedman2001multivariate, slonim2006multivariate, slonim2005estimating, bekkerman2006combinatorial}. 
    However, McGill's seminal work on information transmission also introduced an entirely separate information-theoretic quantity under the name multivariate information.  This quantity's name shifted over time to also being called multi-information and multivariate mutual information, as referenced in \cite{vergara2014review,jakulin2003quantifying}.  This naming confusion has been the source of false equalities between the two quantities, such as in \cite[Eqn.~8]{summaryinfo}.  One fundamental difference in the two quantities is that multi-information as defined here is non-negative, and multivariate mutual information can take positive and negative values.

\section{Distributions with fixed Margins}
\label{appendix:transportation_polytopes}

In various contexts one is interested in the sets joint distributions that are compatible with given values of some of the marginals. We encounter a version of this problem here. Concretely, we study whether certain collections of subsets of variables can each have maximum multi-information, and, if so, what is the set of joint distributions that are compatible with these specifications. 

\subsection{Transportation polytopes}

An $n$-table is an $N_1\times\cdots \times N_n$ array of nonnegative real numbers. 
For a set of indices $\lambda\subseteq\{1,\ldots, n\}$, the $\lambda$-margin of a table $p$ is the $|\lambda|$-table obtained by summing the entries over all but the $\lambda$ indices. 
A multi-index transportation polytope or contingency table is the set of all $n$-tables that satisfy a set of given margins. 
The polytope is called \emph{central} if all entries of any of the considered margins are equal. 
An assignment polytope is the special case where all entries across all of the considered margins are equal (which implies $N_1=\cdots=N_n=N$). 
The two-way assignment polytopes are also called Birkhoff polytopes or polytopes of doubly stochastic matrices (when the value of the margin entries is $1$). 
The $N!$ vertices of the Birkhoff polytope are the permutation matrices of size $N\times N$. 

Optimal transportation problems are often formulated as linear programs with feasible set given by a transportation polytope. Hence the number of vertices and edges has been subject of study. 
For a 2-way transportation polytope, it is known that the dimension is equal to $(N_1-1)(N_2-1)$ (if the entry sums of the two margins are equal). 
A matrix $p$ in the polytope is a vertex if and only if the graph $F_p$ with edges $\{(x_1,x_2) \colon p(x_1,x_2)>0 \}$ forms a spanning forest of the full bipartite graph with edges $\{(x_1,x_2)\colon 1\leq x_1\leq N_1, 1\leq x_2\leq N_2\}$. See~\cite{yemelichev1984polytopes}. 
Already for $3$-way transportation polytopes, the number of vertices is not known in general~\cite{JDL}. 

The set of maximizers of the SFMI is a union of polytopes, each of which has the structure of an $n$-way assignment polytope defined by $1$-margins of size $N$.

\subsection{Marginal polytopes}

Marginal distributions satisfy certain relations among each other. Consider the tuple of $\lambda$-marginals for all $\lambda\in \Lambda$ for some fixed $\Lambda\in 2^{\{1,\ldots, n\}}$. The set of all possible tuples is called the marginal polytope. 
If we consider all $q$-marginals among $n$ variables with state spaces of cardinality $N$, the dimension of the marginal polytope is $\sum_{i=1}^q{n\choose i}(N-1)^i$. 
In turn, a fixed feasible and generic choice of all $q$ marginals cuts out a set (a transportation polytope) in the joint probability simplex of dimension $N^n-1 - \sum_{i=1}^q{n\choose i}(N-1)^i$. 
If the choice of the margins sits at the boundary of the marginal polytope, then the set that is cut out in the joint probability simplex will be smaller. 
Classifying the possible values of the dimension is an open problem in general. 
In our discussion of transportation polytopes maximizing the SFMI, we remove zeros and obtain uniform margins. 

Marginalizing over $n-q$ variables can be viewed as a linear projection $\sL : \D_n \to \D_q$. 
Each marginal distribution that we specify in $\D_q$ provides $N^q$ linear constraints (including normalization). 
In our discussion with margins being maximizers of multi-information, $N$ entries of the margin take value $1/N$ and the other $N^q-N$ entries take value $0$. 
The constraints with value $0$ imply that $(N^q-N)\cdot N^{n-q}$ entries of the joint distributions also take value $0$. 
So each margin creates $N$ linear constraints for $N\cdot N^{n-q}$ entries of the joint distribution, and in addition, it equates the remaining %$(N^q-N)\cdot N^{n-q}$ 
entries of the joint distribution to zero. 
If we forget about the inequality constraints, the number of linearly independent constraints that arise from the specification of the marginal distributions is equal to the dimension of the of the marginal polytope. 

\section{Examples of the sets of SFMI maximizers}

\subsection{The dimension of the sets of SFMI maximizers}

The linear sets of SFMI maximizers come from taking the intersection of $n$ margin constraints, each of which defines a set $\sS_{i,j}$.
Each of these sets is defined by a number of linear equality constraints, which in general are not independent for multiple margins. 
When there are $n$ pairs of binary variables, each $\sS_{i,j}$ is supported on $2^{2n-1}$ vectors. 
Only half of the support vectors between any given $\sS_{i,j}$ agree, meaning $\sS_{i,j} \cap \sS_{i,j}'$ has common support of size $2^{2n-2}$, and $\sS_{i,j} \cap \sS_{i,j}' \cap \sS_{i,j}''$ has support of size $2^{2n-3}$. 
Once the intersection of $n$ of these sets is performed, one is left with $2^{2n-n} = 2^n$ support vectors. 
Then, the linearly independent constraints arising from the $\sS$'s can be subtracted off of $2^n$. 
The result will be the dimension of the linear sets. 
The sets of maximizers also need to satisfy linear inequality constraints. In general, these can further reduce the dimension of the solution set. 

\begin{example}
Consider $n = 1$ pair of $N=2$ valued (binary) variables.  The SFMI maximizing set is exactly the MI maximizers, i.e. a zero dimensional set.  The set $\sS_{1,1}$ has support on $2^{1} = 2$ vectors, $00$ and $11,$ and there are two constraints, $a = 1/2$ and $b = 1/2$ for $a\d_{00} + b\d_{11}.$ Written as a matrix, it would be a rank 2 matrix of constraints.  Normalization is satisfied by these constraints, and the result is a zero dimensional set.
\end{example}

\begin{example}
Consider $n=2$ pairs of $N=2$ valued variables. Then $\sS_{1,1} \cap \sS_{2,1}$ has $4$ vectors in the support, and the constraints that arise are given by 
\bq \bpm
1 & 1 & 0 & 0\\
0 & 0 & 1 & 1\\
1 & 0 & 1 & 0\\
0 & 1 & 0 & 1 \epm 
\bpm a\\b\\c\\d \epm = \bpm 1\\1\\1\\1 \epm,
\eq
where $p = \frac{a}{2}\d_{0000} + \frac{b}{2}\d_{0101} + \frac{c}{2}\d_{1010} + \frac{d}{2}\d_{1111}.$  This matrix is rank 3, meaning the dimension of the resulting linear set is $4-3 = 1$. 
\end{example}

\begin{example} \label{example:28}
When $n=3$ pairs of $N=2$ valued variables.  
Then the intersection $\sS_{1,1} \cap \sS_{2,1} \cap \sS_{3,1}$ has support on $2^3 = 8$ vectors, and the rank of the constraint matrix,
\bq \bpm 
        1 & 1 & 1 & 1 & 0 & 0 & 0 & 0\\
        0 & 0 & 0 & 0 & 1 & 1 & 1 & 1\\
        1 & 1 & 0 & 0 & 1 & 1 & 0 & 0\\
        0 & 0 & 1 & 1 & 0 & 0 & 1 & 1\\
        1 & 0 & 1 & 0 & 1 & 0 & 1 & 0\\
        0 & 1 & 0 & 1 & 0 & 1 & 0 & 1
    \epm \bpm a\\b\\c\\d\\e\\f\\g\\h \epm = \bpm 1\\1\\1\\1\\1\\1 \epm 
\eq
is 4. Therefore we expect the linear sets of maximizers to be of dimension of null space, $8-4 = 4$. 
\end{example}

\begin{example}
Consider $n=2$ pairs of $N=3$ valued variables. Then $\sS_{1,1} \cap \sS_{2,1}$ is supported on $N^n = 9$ vectors.  The constraint matrix looks like,
\bq \bpm 
1 &1 &1 &0 &0 &0 &0 &0 &0\\
0 &0 &0 &1 &1 &1 &0 &0 &0\\
0 &0 &0 &0 &0 &0 &1 &1 &1\\
1 &0 &0 &1 &0 &0 &1 &0 &0\\
0 &1 &0 &0 &1 &0 &0 &1 &0\\
0 &0 &1 &0 &0 &1 &0 &0 &1 \epm \bpm a\\b\\c\\d\\e\\f\\g\\h\\u \epm = \bpm 1\\1\\1\\1\\1\\1\\1\\1\\1 \epm.
\eq

Here, the coefficients correspond to the joint distributions probabilities on $0000, 0101, 02020, 1010, 1111, 1212, 2020, 2121, 2222$ respectively.  This matrix is rank $5$ meaning the null space is $4$ dimensional.
\end{example}

\subsection{Transportation polytopes of SFMI maximizers}
\label{sec:sfmiex}

\begin{example} Consider the case of $n=2$ pairs of $N=2$ valued variables. 
In the following we illustrate explicitly how the SFMI maximizing distributions noted in Example~\ref{ex:1} are all of the SFMI maximizers. 
Since it is known that distributions maximizing SFMI must maximize each term of the sum individually, one may consider the set of SFMI maximizers as the intersection of $n$ sets where each set contains all the distributions which maximize a single term of the sum, i.e. each maximize $D(p(x_i,y_i)\|\sF)$ for some $i$. 

Consider the four binary variables $X_1,X_2,Y_1,Y_2$. 
The maximizers of $I(X_1,Y_1)$ are $p^1_1(X_1,Y_1)=\frac12( \delta_{00} +\delta_{11})$ and $p^2_1(X_1,Y_1)=\frac12( \delta_{01} +\delta_{10})$.  The corresponding joint distributions have the form $p(X_1,X_2,Y_1,Y_2) = p^j_1(X_1,Y_1)q(X_2,Y_2|X_1,Y_1)$.  Therefore, one may define $\sS_{1,1}$ to be the set of distributions which maximize $I(X_1,Y_1)$ and have marginal $p^1_1(X_1,Y_1).$  Likewise, $\sS_{1,2}$ is the set of distributions which maximize $I(X_1,Y_1)$ and have marginal $p^2_1(X_1,Y_1).$  These can be written as 
    \bq
    \label{eq:14} \bigg\{p \in \D_{15} \; : \; p = \frac12\sum_{x_2,y_2}q(x_2,y_2|0,0)\delta_{0x_20y_2} + \frac12\sum_{x_2,y_2}q(x_2,y_2|1,1)\delta_{1x_21y_2}\bigg\}
    \eq
    and
    \bq \label{eq:15} \bigg\{p \in \D_{15} \; : \; p = \frac12\sum_{x_2,y_2}q(x_2,y_2|0,1)\delta_{0x_21y_2} + \frac12\sum_{x_2,y_2}q(x_2,y_2|1,0)\delta_{1x_20y_2}\bigg\}
    \eq
    respectively.
    Here, $q(x_2,y_2|x_1=0,y_1=0)$ is the conditional probability of observing the joint state $(X_1 = 0, X_2 = x_2, Y_1 = 0, Y_2 = y_2).$  Similarly, the two sets $\sS_{2,1}$ and $\sS_{2,2}$ which maximize $I(X_2,Y_2)$ and have marginals $p^1_2(x_2,y_2)=\frac12( \d_{00} +\d_{11})$ and $p^2_2(x_2,y_2)=\frac12( \d_{01} +\d_{10})$ respectively can be written as
    \bq
    \label{eq:16} \bigg\{p \in \D_{15} \; : \; p = \frac12\sum_{x_1,y_1}q(x_1,y_1|0,0)\delta_{x_10y_10} + \frac12\sum_{x_1,y_1}q(x_1,y_1|1,1)\delta_{x_11y_11}\bigg\}
    \eq
    and
    \bq\label{eq:17} \bigg\{p \in \D_{15} \; : \; p = \frac12\sum_{x_1,y_1}q(x_1,y_1|0,1)\delta_{x_10y_11} + \frac12\sum_{x_1,y_1}q(x_1,y_1|1,0)\delta_{x_11y_10}\bigg\}.
    \eq

    There are 4 possibilities for how these sets may be intersected to yield maximizers of SFMI.  These are, $\sS_{1,1} \cap \sS_{2,1},$ $\sS_{1,1} \cap \sS_{2,2},$ $\sS_{1,2} \cap \sS_{2,1},$ and $\sS_{1,2} \cap \sS_{2,2}.$  It will be shown that the intersection of $\sS_{1,1} \cap \sS_{2,1}$ corresponds to the line of distributions in Eqn. (\ref{eq:10}), and the other intersections correspond in a similar fashion.  To see this, one may begin by considering the binary codes in the intersection.  These are all the binary codes which match in the 1st and 3rd entry, and match in the 2nd and 4th entry.  Thus, the distributions in $\sS_{1,1} \cap \sS_{2,1}$ can only have support on $\{0000, 0101, 1010, 1111\}$.  Thus, at most the intersection could be a 3 dimensional set, due to the normalization constraint.  Next, consider the linear constraints build into $\sS_{1,1}$.  Letting 
    \begin{align}
    \begin{split}
%    \label{eq:cond1}
    a &= q(X_2 = 0, Y_2 = 0 | X_1 = 0, Y_1 = 0)\\
    b &= q(X_2 = 0, Y_2 = 0 | X_1 = 1, Y_1 = 1)\\
    c &= q(X_2 = 1, Y_2 = 1 | X_1 = 0, Y_1 = 0)\\
 %   \label{eq:cond2}
    d &= q(X_2 = 1, Y_2 = 1 | X_1 = 1, Y_1 = 1),
    \end{split}
    \end{align}
    the constraints can be written as $a+c = 1$ and $b+d = 1$. 
    Similarly, the distribution in $\sS_{2,1}$ must satisfy $a+b=1$ and $c+d=1.$  These constraints require that $a = d,$ and $b = c.$  The distributions in $\sS_{1,1} \cap \sS_{2,1}$ must take the form 
    \bq
    p = \frac{a}{2} \d_{0000} + \frac{b}{2} \d_{1010} + \frac{c}{2} \d_{0101} + \frac{d}{2} \d_{1111}.
    \eq
      Letting $\a = \frac{a}{2} = \frac{d}{2}$ and $\b = \frac{b}{2} = \frac{c}{2},$ the distributions in $\sS_{1,1} \cap \sS_{2,1}$ can be written as in Eqn. (\ref{eq:10}),
    \bq
    p = \a(\d_{0000} + \d_{1111}) + \b(\d_{0101} + \d_{1010}),
    \eq
    where $\a + \b = \frac{1}{2}$.  By considering the other intersections, $\sS_{1,1} \cap \sS_{2,2}, \; \sS_{1,2} \cap \sS_{2,1}, \; \sS_{1,2} \cap \sS_{2,2},$ the other 3 transportation polytopes in Example~\ref{ex:1} are recovered.
\end{example} 

\begin{example}
\label{ex:2}
    Consider the case of $n=3$ pairs of $N=2$ valued variables $(X_1,X_2,X_3,Y_1,Y_2,Y_3)$. 
    For any $(i,j) \in \{1,2,3\}\times\{1,2\}$, let $\sS_{i,j}$ denote the set of joint distributions whose $(X_i,Y_i)$ margin is the $j$-th maximizer of the mutual information. 
    Each $\sS_{i,j}$ contains distributions supported on $2^{6-1} = 32$ binary vectors. 
    The intersections $\sS_{1,j_1}\cap\sS_{2,j_2}$ contain distributions supported on $2^{6-2} = 16$ vectors, 
    and when intersecting $\sS_{1,j_1} \cap \sS_{2,j_2} \cap \sS_{3,j_3}$, we obtain distributions supported on $8$ binary vectors. 
    
    For concreteness, let $\sS_{i,1}$ correspond to $p(X_i,Y_i) = \frac12(\d_{00} + \d_{11})$. 
    Then $\sS_{1,1}$ is the set of joint distributions
    of the form 
    \begin{multline}
p = \frac{1}{2}\sum\limits_{\substack{x_2,x_3 \\ y_2,y_3}}\Big(q(x_2,x_3,y_2,y_3|0,0)\delta_{0x_2x_30y_2y_3} + q(x_2,x_3,y_2,y_3|1,1)\delta_{1x_2x_31y_2y_3} \Big),  
\label{eq:firstway}
    \end{multline}
where the numbers $q(x_2,x_3,y_2,y_3|x_1,y_1)$ define an arbitrary conditional distribution. 
Each of these $p$ is supported on strings $(x_1,x_2,x_3,y_1,y_2,y_3)$ that satisfy $x_1=y_1$. 

The distributions $p \in \sS_{1,1} \cap \sS_{2,1} \cap \sS_{3,1}$ are supported on strings $(x_1,x_2,x_3,y_1,y_2,y_3)$ that satisfy $x_i=y_i$ for all $i=1,2,3$. 
In turn, they have the form
{\small
    \begin{equation*} \label{eq:32}
    p = \frac{a}{2}\d_{000000} + \frac{b}{2}\d_{001001} + \frac{c}{2}\d_{010010} +
    \frac{d}{2}\d_{011011} + \frac{e}{2}\d_{100100} +  \frac{f}{2}\d_{101101} + \frac{g}{2}\d_{110110} + \frac{h}{2}\d_{111111}.
    \end{equation*}
}

The coefficients correspond to conditional probability distributions in three different ways. 
One way is as the conditional probabilities of $X_2,X_3,Y_2,Y_3$ given $X_1,Y_2$, with 
    \begin{align}
    \begin{split}
    a = q(0,0,0,0 | 0,0),\qquad & e = q(0,0,0,0 | 1,1),\\
    b = q(0,1,0,1 | 0,0),\qquad & f = q(0,1,0,1 | 1,1),\\
    c = q(1,0,1,0 | 0,0),\qquad & g = q(1,0,1,0 | 1,1),\\
    d = q(1,1,1,1 | 0,0),\qquad & h = q(1,1,1,1 | 1,1), 
    \end{split}
    \end{align} 
which corresponds to~\eqref{eq:firstway} and implies that $a+b+c+d=1$ and $e+f+g+h=1$. 
The other two ways result from conditioning on $X_2,Y_2$ and $X_3,Y_3$, respectively, and imply two more pairs of linear constraints. 
The six constraints form the linear system $M{\bf q} = {\bf 1}$ with
    \bq \bpm 
        1 & 1 & 1 & 1 & 0 & 0 & 0 & 0\\
        0 & 0 & 0 & 0 & 1 & 1 & 1 & 1\\
        1 & 1 & 0 & 0 & 1 & 1 & 0 & 0\\
        0 & 0 & 1 & 1 & 0 & 0 & 1 & 1\\
        1 & 0 & 1 & 0 & 1 & 0 & 1 & 0\\
        0 & 1 & 0 & 1 & 0 & 1 & 0 & 1
    \epm \bpm a\\b\\c\\d\\e\\f\\g\\h \epm = \bpm 1\\1\\1\\1\\1\\1 \epm .
    \eq
    The matrix $M$ is simply the map computing the $(X_i,Y_i)$-margins, restricted to the support points (i.e., strings with $x_i=y_i$ for $i=1,2,3$). 
    It can be regarded as a sufficient statistics matrix of a binary independence model of $n$ variables, including the statistic that computes the total mass. % (normalization). 
    The matrix has rank $4$ and the vector $\frac14 {\bf 1} \in \R^8$ is a particular solution of the linear system.  At this particular solution, one obtains a maximizer of $MI({\bf X},{\bf Y}).$ 
    
    As discussed in Example~\ref{example:28}, the space of solutions is 4 dimensional.  To obtain all solutions, we consider the kernel of $M$, which is spanned by the rows of 
    \begin{equation}
    \begin{pmatrix}
     1&  1& -1& -1& -1& -1&  1&  1\\
     1& -1&  1& -1& -1&  1& -1&  1\\
     1& -1& -1&  1&  1& -1& -1&  1\\
     1& -1& -1&  1& -1&  1&  1& -1
     \end{pmatrix}.
     \label{eq:exkernelM}
    \end{equation}
    Any linear combination of these rows can be added to $\frac14 {\bf 1} \in \R^8$ to obtain another vector that satisfies the marginal equality constraints. 
    The intersection of the $4$-dimensional affine space with the probability simplex is obtained by requiring that the solutions have non-negative entries. 
    This results in a polytope. We compute the vertex representation using the free open-source mathematics software Sage~\cite{sagemath}, and obtain the $6$ rows 
    \begin{equation}
    \begin{pmatrix}
    0& 0& 0& 1& 1& 0& 0& 0 \\
    0& 0& 1& 0& 0& 1& 0& 0 \\
    0& 1& 0& 0& 0& 0& 1& 0 \\
    1& 0& 0& 0& 0& 0& 0& 1 \\
    1/2& 0& 0& 1/2& 0& 1/2& 1/2& 0 \\
    0& 1/2& 1/2& 0& 1/2& 0& 0& 1/2 
    \end{pmatrix}. 
    \end{equation}
    
    \begin{comment}
    % COMPUTE THE VERTEX PRESENTATION OF THE POLYTOPE USING SageMathCell https://sagecell.sagemath.org
P = Polyhedron(ieqs=[(-1, 1, 1, 1, 1, 0, 0, 0, 0),
                    (-1, 0, 0, 0, 0, 1, 1, 1, 1),
                    (-1, 1, 1, 0, 0, 1, 1, 0, 0),
                    (-1, 0, 0, 1, 1, 0, 0, 1, 1),
                    (-1, 1, 0, 1, 0, 1, 0, 1, 0), 
                    (-1, 0, 1, 0, 1, 0, 1, 0, 1),
                    (1, -1, -1, -1, -1, 0, 0, 0, 0),
                    (1, 0, 0, 0, 0, -1, -1, -1, -1),
                    (1, -1, -1, 0, 0, -1, -1, 0, 0),
                    (1, 0, 0, -1, -1, 0, 0, -1, -1),
                    (1, -1, 0, -1, 0, -1, 0, -1, 0), 
                    (1, 0, -1, 0, -1, 0, -1, 0, -1), 
                    (0, 1, 0, 0, 0, 0, 0, 0, 0),
                    (0, 0, 1, 0, 0, 0, 0, 0, 0),
                    (0, 0, 0, 1, 0, 0, 0, 0, 0),
                    (0, 0, 0, 0, 1, 0, 0, 0, 0),
                    (0, 0, 0, 0, 0, 1, 0, 0, 0),
                    (0, 0, 0, 0, 0, 0, 1, 0, 0),
                    (0, 0, 0, 0, 0, 0, 0, 1, 0),
                    (0, 0, 0, 0, 0, 0, 0, 0, 1)])
P
% returns 
% A 4-dimensional polyhedron in QQ^8 defined as the convex hull of 6 vertices (use the .plot() method to plot)

P.vertices()
% returns 
% (A vertex at (0, 0, 0, 1, 1, 0, 0, 0),
% A vertex at (0, 0, 1, 0, 0, 1, 0, 0),
% A vertex at (0, 1, 0, 0, 0, 0, 1, 0),
% A vertex at (1/2, 0, 0, 1/2, 0, 1/2, 1/2, 0),
% A vertex at (1, 0, 0, 0, 0, 0, 0, 1),
% A vertex at (0, 1/2, 1/2, 0, 1/2, 0, 0, 1/2))
    \end{comment}
    
 Based off of the vertex representation above, the distributions which are in $\sS = \sS_{1,1} \cap \sS_{2,1} \cap \sS_{3,1}$ are of the form 
 \begin{multline*}
p = 
 \tfrac{\a_1}{2}(\d_{000000} + \d_{111111}) + \tfrac{\a_2}{2}(\d_{100100} + \d_{011011}) \\
+\tfrac{\a_3}{2}(\d_{010010} + \d_{101101}) 
+\tfrac{\a_4}{2}(\d_{001001} + \d_{110110}) \\ 
+\tfrac{\a_5}{4}(\d_{000000} + \d_{011011} + \d_{101101} + \d_{110110})\\
+\tfrac{\a_6}{4}(\d_{001001}  + \d_{010010}  + \d_{100100}  + \d_{111111}), \end{multline*}
     where %$2(\a_1 + \a_2 + \a_3 + \a_4) + 4(\a_5 + \a_6) = 1$
     $\sum_i \a_i = 1$ ensures normalization and $\a_i \geq 0$ for $i = 1,\dots,6$. 
     One can quickly check that $p(X_i,Y_i) = \frac12(\d_{00} + \d_{11})$ for $i = 1,2,3$. 
    The vertices $\alpha_1=1$, $\alpha_2=1$, $\alpha_3=1$, $\alpha_4=1$ are the only distributions with support of cardinality two and they are also the only maximizers of multi-information within $\sS$. 
    The choices $\a_1 = \a_2 = \a_3 = \a_4=1/4$ and $\a_5 = \a_6=1/2$ give the same distribution, which is also the only maximizer of $MI({\bf X},{\bf Y})$ within $\sS$. 
    This observation motivates Theorem~\ref{theorem:13} which states that each polytope of SFMI maximizers contains one or more simplices, and at the center of each simplex lies a $MI({\bf X},{\bf Y})$ maximizer.  Furthermore, when restricted to the set of SFMI maximizers, the distributions with maximum entropy are $MI({\bf X},{\bf Y})$ maximizers. 
    
    We have discussed one of the $N!^n = 2!^3 = 8$ polytopes. The other $7$ are similar. 
    Each of them has $4$ vertices that are maximizers of the multi-information and form a $3$-simplex 
    with a maximizer of $MI({\bf X},{\bf Y})$ at its center. 
    The $8$ centers are following vectors scaled by $1/8$, whereby grouped terms correspond to the 
    vertices: 
     {\small
     \begin{align*}
     %\sS_{1}^1 \cap \sS_{2}^1 \cap \sS_{3}^1: 
     %p &= \a_1
     &(\d_{000000} + \d_{111111}) + %\a_2
     (\d_{001001} + \d_{110110}) + %\a_3
     (\d_{010010} + \d_{101101}) + %\a_4
     (\d_{011011} + \d_{100100}),\\
     %\sS_{1}^1 \cap \sS_{2}^1 \cap \sS_{3}^2: 
     %p &= \a_1
     &(\d_{000001} + \d_{111110}) + %\a_2
     (\d_{001000} + \d_{110111}) + %\a_3
     (\d_{010011} + \d_{101100}) + %\a_4
     (\d_{011010} + \d_{100101}),\\
     %\sS_{1}^1 \cap \sS_{2}^2 \cap \sS_{3}^1: 
     %p &= \a_1
     &(\d_{000010} + \d_{111101}) + %\a_2
     (\d_{010000} + \d_{101111}) + %\a_3
     (\d_{001011} + \d_{110100}) + %\a_4
     (\d_{011001} + \d_{100110}),\\
     %\sS_{1}^2 \cap \sS_{2}^1 \cap \sS_{3}^1: 
     %p &= \a_1
     &(\d_{000100} + \d_{111011}) + %\a_2
     (\d_{100000} + \d_{011111}) + %\a_3
     (\d_{001101} + \d_{110010}) + %\a_4
     (\d_{101001} + \d_{010110}),\\
     %\sS_{1}^1 \cap \sS_{2}^2 \cap \sS_{3}^2: 
     %p &= \a_1
     &(\d_{000011} + \d_{111100}) + %\a_2
     (\d_{001010} + \d_{110101}) + %\a_3
     (\d_{010001} + \d_{101110}) + %\a_4
     (\d_{011000} + \d_{100111}),\\
     %\sS_{1}^2 \cap \sS_{2}^1 \cap \sS_{3}^2: 
     %p &= \a_1
     &(\d_{000101} + \d_{111010}) + %\a_2
     (\d_{001100} + \d_{110011}) + %\a_3
     (\d_{100001} + \d_{011110}) + %\a_4
     (\d_{101000} + \d_{010111}),\\
     %\sS_{1}^2 \cap \sS_{2}^2 \cap \sS_{3}^1: 
     %p &= \a_1
     &(\d_{000110} + \d_{111001}) + %\a_2
     (\d_{010100} + \d_{101011}) + %\a_3
     (\d_{100010} + \d_{011101}) + %\a_4
     (\d_{110000} + \d_{001111}),\\
     %\sS_{1}^2 \cap \sS_{2}^2 \cap \sS_{3}^2: 
     %p &= \a_1
     &(\d_{000111} + \d_{111000}) + %\a_2
     (\d_{001110} + \d_{110001}) + %\a_3
     (\d_{010101} + \d_{101010}) + %\a_4
     (\d_{011100} + \d_{100011}).
     \end{align*}}
    One can see that the simplices are disjoint. 
    In fact, distributions from different simplices have disjoint supports, meaning that the simplices sit in disjoint faces of the joint probability simplex. 
    These simplices having disjoint support only occurs when dealing with binary variables.
\end{example}   

\begin{example}
\label{ex:19}
Consider the case of $n=2$ pairs of $N=3$ valued variables. 
For each $i=1,2$ let $\sS_{i,{j_i}}$ be the set of joint distributions which have marginal $p(X_i,Y_i)$ as one of the $N! = 3!=6$ multi-information maximizers, so $j_i \in \{1,\dots,6\}$. 
For concreteness, let the correspondence be
\begin{align*}
\sS_{i,1} :&\quad p(X_i,Y_i) = \frac13(\d_{00} + \d_{11} + \d_{22}),\\
\sS_{i,2} :&\quad p(X_i,Y_i) = \frac13(\d_{00} + \d_{12} + \d_{21}),\\
\sS_{i,3} :&\quad p(X_i,Y_i) = \frac13(\d_{01} + \d_{10} + \d_{22}),\\
\sS_{i,4} :&\quad p(X_i,Y_i) = \frac13(\d_{01} + \d_{12} + \d_{20}),\\
\sS_{i,5} :&\quad p(X_i,Y_i) = \frac13(\d_{02} + \d_{11} + \d_{20}),\\
\sS_{i,6} :&\quad p(X_i,Y_i) = \frac13(\d_{02} + \d_{10} + \d_{22}).
\end{align*}
There are $N!^n = 3!^2 = 36$ intersections of the form $\sS_{1,{j_1}} \cap \sS_{2,{j_2}}$, each corresponding to a transportation polytope. 
They together contain all $N!^{2n-1}=3!^{2\cdot 2-1}=216$ maximizers of the multi-information. 
Consider the
support sets of three to start: 
\begin{align*}
\supp(\sS_{1,1} \cap \sS_{2,1}) &= \{0000,0101,0202,1010,1111,1212,2020,2121,2222 \},\\
\supp(\sS_{1,1} \cap \sS_{2,2}) &= \{0000,0102,0201,1010,1112,1211,2020,2122,2221 \},\\
\supp(\sS_{1,1} \cap \sS_{2,3}) &= \{0001,0100,0202,1011,1110,1212,2021,2120,2222 \}. 
\end{align*}
The probabilities assigned to the joint states are not arbitrary, since they need to satisfy the margins. 
One particular subset of distributions in $\sS_{1,1} \cap \sS_{2,1}$ can be written as 
\begin{equation*}
p = \frac{\a_1}{3}(\d_{0000} + \d_{1111} + \d_{2222}) + \frac{\a_2}{3}(\d_{0101} + \d_{1212} + \d_{2020}) + \frac{\a_3}{3}(\d_{0202} + \d_{1010} + \d_{2121}),
\end{equation*}
where $\a_1 + \a_2 + \a_3 = 1$. 
This is a simplex whose $N^{n-1}=3^{2-1}=3$ vertices maximize multi-information. 
Another one can be written as 
\begin{equation*}
p = \frac{\a_1}{3}(\d_{0000} + \d_{1212} + \d_{2121}) + \frac{\a_2}{3}(\d_{1111} + \d_{2020} + \d_{0202}) + \frac{\a_3}{3}(\d_{2222} + \d_{0101} + \d_{1010}). 
\end{equation*}
This is another simplex whose $3$ vertices maximize multi-information. 
Both simplices have the same centroid, with $\alpha_1=\alpha_2=\alpha_3=1/3$, which is a maximizer of $MI({\bf X},{\bf Y})$. 

The other transportation polytopes (intersections $\sS_{1,{j_1}} \cap \sS_{2,{j_2}}$) are similar. 
For example, the two simplices in $\sS_{1,1} \cap \sS_{2,2}$ are 
\begin{align*}
%p &= 
&\frac{\a_1}{3}(\d_{0000} + \d_{1211} + \d_{2122}) +  \frac{\a_2}{3}(\d_{1112} + \d_{2020} + \d_{0201}) +  \frac{\a_3}{3}(\d_{1010} + \d_{0102} + \d_{2221}),\\
%p &=  
&\frac{\a_1}{3}(\d_{0000} + \d_{1112} + \d_{2221}) +  \frac{\a_2}{3}(\d_{1211} + \d_{0102} + \d_{2020}) +  \frac{\a_3}{3}(\d_{2122} + \d_{1010} + \d_{0201}), 
\end{align*}
and the two simplices in $\sS_{1,1} \cap \sS_{2,3}$ are 
\begin{align*}
%p &=  
&\frac{\a_1}{3}(\d_{0001} + \d_{1212} + \d_{2120}) +  \frac{\a_2}{3}(\d_{2222} + \d_{1011} + \d_{0100}) +  \frac{\a_3}{3}(\d_{2021} + \d_{1110} + \d_{0202}),\\
%p &=  
&\frac{\a_1}{3}(\d_{0001} + \d_{2222} + \d_{1110}) +  \frac{\a_2}{3}(\d_{1212} + \d_{2021} + \d_{0100}) +  \frac{\a_3}{3}(\d_{2120} + \d_{0202} + \d_{1011}).
\end{align*}

Observe that there are $N!^{n-1}=3!^{2-1}=6$ multi-information maximizers per transportation polytope, and there are $N!^n$ transportation polytopes total. 
This number is expected because all of the $N!^{2n-1}$ maximizers of multi-information are SFMI maximizers. 
\end{example}

\end{document}